\documentclass{CSML}

\pdfoutput=1

\def\dOi{12(3:4)2016}
\lmcsheading%
{\dOi}
{1--26}
{}
{}
{Oct.~29, 2015}
{Jul.~27, 2017}
{}

\ACMCCS{[{\bf Software and its engineering}]: Software organization and properties---Software functional properties---Correctness\,/\,Formal methods---Software verification}

\usepackage{amsthm}
\usepackage{amsmath}
\usepackage{amssymb}

\usepackage{booktabs}

\usepackage{tikz}
\usetikzlibrary{automata}
\usetikzlibrary{backgrounds}

\tikzset{
  initial text=,
  every path/.style={->,-stealth, every loop/.style={-stealth}},
  every initial by arrow/.style={-stealth, initial text={},
  every loop/.style={red,-stealth}},
}

\usepackage{pifont}
\usepackage{fancyvrb}

\usepackage{hyperref}
\hypersetup{hidelinks}

\newcommand{\noerror}{\text{\ding{51}}}
\newcommand{\copyerror}{\text{\ding{55}}}

\newcommand{\myquot}[1]{``#1''}

\newcommand{\nats}{\mathbb{N}}

\newcommand{\size}[1]{|#1|}
\renewcommand{\epsilon}{\varepsilon}

\newcommand{\proj}[1]{\mathrm{pr}_{\hspace{-.5pt}I}(#1)}
\newcommand{\set}[1]{\{#1\}}

\newcommand{\widx}[2]{#1(#2)}

\newcommand{\aut}{\mathcal{A}}
\newcommand{\autb}{\mathcal{B}}
\newcommand{\autc}{\mathcal{C}}

\newcommand{\tm}{\mathcal M}
\newcommand{\col}{\Omega}

\newcommand{\lstar}[1]{L_*(#1)}
\newcommand{\lexists}[1]{L_\exists(#1)}
\newcommand{\lall}[1]{L_\forall(#1)}
\newcommand{\lpar}[1]{L_p(#1)}

\newcommand{\acc}{\mathrm{acc}}

\newcommand{\delaygame}[1]{\Gamma\!_{f}(#1)}
\newcommand{\delaygamep}[1]{\Gamma\!_{f'}(#1)}
\newcommand{\SigmaI}{\Sigma_I}
\newcommand{\SigmaO}{\Sigma_O}
\newcommand{\strat}{\tau}
\newcommand{\stratO}{\tau_O}
\newcommand{\stratI}{\tau_I}
\newcommand{\p}{P}

\newcommand{\conp}{\textsc{co-NP}}
\newcommand{\pspace}{\textsc{PSpace}}
\newcommand{\apspace}{\textsc{APSpace}}
\newcommand{\exptime}{\textsc{ExpTime}}
\newcommand{\bigo}{\mathcal{O}}
\newcommand{\threeexp}{\textsc{3ExpTime}}

\newcommand{\autp}{\mathcal{P}}
\newcommand{\resolve}{r}
\newcommand{\game}{\mathcal{G}}
\newcommand{\wit}[1]{W_{#1}}
\newcommand{\init}{v_I}
\newcommand{\curlyR}{\mathfrak{R}}
\newcommand{\dom}{\mathrm{dom}}
\newcommand{\block}[1]{\overline{#1}}

\newcommand{\qacc}{q_{A}}
\newcommand{\qrej}{q_{R}}

\SaveVerb{sharp}|#|

\newcommand{\sharpsym}{\mathop{\protect\raisebox{-0.5pt}{\protect\scalebox{1.15}{\protect\UseVerb{sharp}}}}}
\newcommand{\autbsharp}{\autb_{ \protect\scalebox{0.8}{\ensuremath{\sharpsym}}}}

\newcommand{\dollar}{\$}

\newcommand{\bin}[1]{\langle #1 \rangle_{\!_2}}

\newcommand{\fresh}{N}
\newcommand{\ccopy}{C}
\newcommand{\conf}{\mathrm{Conf}}

\begin{document}

\title[How Much Lookahead is Needed to Win Infinite Games?]
      {How Much Lookahead is Needed to Win Infinite Games?\rsuper*}
\titlecomment{{\lsuper*}A preliminary version of this work appeared in the proceedings of ICALP 2015~\cite{KleinZimmermann15a}}

\author[F.~Klein]{Felix Klein\rsuper a}	
\address{{\lsuper {a,b}}Reactive Systems Group, Saarland University, 66123 Saarbrücken, Germany}	
\email{\{klein,zimmermann\}@react.uni-saarland.de}  
\thanks{{\lsuper a}Supported by the Transregional Collaborative Research Center~\myquot{AVACS} (SFB/TR 14) of the German Research Foundation (DFG) and by an IMPRS-CS PhD scholarship.}	

\author[Martin Zimmermann]{Martin Zimmermann\rsuper b}	
\address{\vspace{-18 pt}}
\thanks{{\lsuper b}Supported by the DFG project~\myquot{TriCS} (ZI 1516/1-1)}	




\begin{abstract}
Delay games are two-player games of infinite duration in
which one player may delay her moves to obtain a lookahead on her
opponent's moves. For $\omega$-regular winning conditions it is known that
such games can be solved in doubly-exponential time and that
doubly-exponential lookahead is sufficient.

We improve upon both results by giving an exponential time algorithm
and an exponential upper bound on the necessary lookahead. This is
complemented by showing $\exptime$-hardness of the solution problem and tight exponential lower bounds on the lookahead. Both lower bounds already hold for safety conditions. Furthermore, solving delay games with reachability conditions is shown to be $\pspace$-complete.

This is a corrected version of the paper \url{https://arxiv.org/abs/1412.3701v4} published originally on August 26, 2016.
\end{abstract}

\maketitle


\section{Introduction}
\label{intro}
Many of today's problems in computer science are no longer concerned with
programs that transform data and then terminate, but with non-terminating
reactive systems which have to interact with a possibly antagonistic
environment for an unbounded amount of time. The framework of infinite
two-player games is a powerful and flexible tool to verify and synthesize such
systems. The seminal theorem of Büchi and Landweber~\cite{BuechiLandweber69}
states that the winner of an infinite game on a finite arena with an $\omega$-regular
winning condition can be determined and a corresponding finite-state winning
strategy can be constructed effectively.

\subsection{Delay Games.} In this work, we consider an extension of the classical framework: in a delay game, one player can postpone her moves for some time to obtain a lookahead on her opponent's moves. This allows her to
win some games which she would loose without lookahead, e.g., if her first move
depends on the third move of her opponent. Nevertheless, there are
winning conditions that cannot be won with any finite lookahead, e.g., if her
first move depends on every move of her opponent. Delay
arises naturally if transmission of data in networks or components equipped
with buffers are modeled.

From a more theoretical point of view,
uniformization of relations by continuous functions~\cite{DBLP:conf/rex/ThomasL93,trakhtenbrot1973finite} can be expressed and
analyzed using delay games~\cite{HoltmannKaiserThomas12,Thomas09}. We consider games in which two players pick letters from alphabets~$\SigmaI$
and $\SigmaO$, respectively, thereby producing two infinite sequences~$\alpha$
and $\beta$. Thus, a strategy for the second player induces a mapping
$\tau\colon \SigmaI^\omega\rightarrow\SigmaO^\omega$. It is winning for
her if $(\alpha,\tau(\alpha))$ is contained in the winning
condition~$L\subseteq \SigmaI^\omega\times \SigmaO^\omega$ for every $\alpha$.
If this is the case, we say that $\tau$ uniformizes $L$. In the classical setting, in
which the players pick letters in alternation, the $n$-th letter of
$\tau(\alpha)$ depends only on the first $n$ letters of $\alpha$. A strategy
with bounded lookahead, i.e., only finitely many
moves are postponed, induces a Lipschitz-continuous function~$\tau$ (in the
Cantor topology on $\Sigma^\omega$) and a strategy with unbounded lookahead
induces a continuous function  (equivalently, a uniformly continuous
function, as $\Sigma^\omega$ is compact (see, e.g.,~\cite{HoltmannKaiserThomas12})).

\subsection{Related Work.} Hosch and Landweber proved that it is decidable whether a delay game with an $\omega$-regular
winning condition can be won with bounded lookahead~\cite{HoschLandweber72}. Later, Holtmann, Kaiser, and Thomas revisited the problem and showed that if the delaying
player wins such a game with unbounded lookahead, then she already wins it with
doubly-exponential bounded lookahead, and gave a streamlined decidability
proof yielding an algorithm with doubly-exponential running time~\cite{HoltmannKaiserThomas12}. Thus, the delaying player does not gain additional power from having unbounded lookahead, bounded lookahead is sufficient.

Going beyond $\omega$-regularity by considering context-free winning conditions leads to undecidability and non-elementary lower bounds on the necessary lookahead, even for very weak fragments~\cite{FridmanLoedingZimmermann11}. Nevertheless, there is another extension of the $\omega$-regular winning conditions where one can prove the analogue of the Hosch-Landweber Theorem: it is decidable whether the delaying player wins a delay game with bounded lookahead, if the winning condition is definable in weak monadic second order logic with the unbounding quantifier (WMSO$+$U)~\cite{Zimmermann15}. Furthermore, doubly-exponential lookahead is sufficient for such winning conditions, provided the delaying player wins with bounded lookahead at all. However, bounded lookahead is not always sufficient to win such games, i.e., the analogue of the Holtmann-Kaiser-Thomas Theorem does not hold for WMSO$+$U winning conditions. In contrast, if the delaying player wins such a game, then she wins no matter how slowly the lookahead grows~\cite{Zimmermann15d}. 

Due to these negative results for WMSO$+$U winning conditions, delay games with winning conditions in Prompt-LTL~\cite{KupfermanPitermanVardi09} have been considered: this logic extends LTL~\cite{Pnueli77} with temporal operators whose scope is bounded in time. Determining the winner of such a game is $\threeexp$-complete and triply-exponential bounded lookahead is necessary in general and always sufficient~\cite{KleinZimmermann16}. Furthermore, all lower bounds already hold for the special case of LTL. 
 
Finally, all delay games with Borel winning conditions are determined,
since they were shown to be reducible to delay-free games with Borel
winning conditions~\cite{KleinZimmermann15}.

Stated in terms of uniformization, Hosch and Landweber proved decidability of the
uniformization problem for $\omega$-regular relations by Lipschitz-continuous functions
and Holtmann et al.\ proved the equivalence of the existence of a
continuous uniformization function and the existence of a Lipschitz-continuous
uniformization function for $\omega$-regular relations. Furthermore, uniformization of context-free relations is undecidable, even with respect to Lipschitz-continuous functions, but uniformization of WMSO$+$U relations by Lipschitz-continuous functions is decidable.

In another line of work, Carayol and Löding considered the case of finite words~\cite{CarayolLoeding12}, and Löding and Winter~\cite{LoedingWinter14} considered the case of finite trees, which are both decidable. However, the non-existence of MSO-definable choice functions on the infinite
binary tree~\cite{CarayolLoeding07,GS83} implies that uniformization fails for such trees.

Although several extensions of $\omega$-regular winning conditions for delay games have been studied, several gaps remain open, even for $\omega$-regular winning conditions. Holtmann et al.\ proved that such games can be solved in doubly-exponential time, but only trivial lower bounds are known. Similarly, their algorithm yields a doubly-exponential upper bound on the necessary lookahead, but only straightforward linear lower bounds are known.  Also, only deterministic parity automata were used to specify winning conditions, and the necessary lookahead and the solution complexity are measured in their size. It is open whether weaker automata models like reachability or safety automata have smaller lookahead requirements and allow for faster algorithms.

\subsection{Our Contribution.} We close these gaps and improve upon both results of Holtmann et al.\ by determining the exact complexity of $\omega$-regular delay games and by giving tight bounds on the necessary lookahead, both for winning conditions specified by deterministic parity automata and for weaker models. 

First, we  present an exponential time algorithm for solving delay games with $\omega$-regular winning conditions, an exponential improvement over the original doubly-exponential time algorithm. Both algorithms share some similarities: given a deterministic parity automaton~$\aut$ recognizing the winning condition of the game, a parity game is constructed that is won by the delaying player if and only if she wins the delay game with winning condition~$L(\aut)$. Furthermore, both parity games are induced by equivalence relations that capture the behavior of $\aut$. However, our parity game is of exponential size while the one of Holtmann et al.\ is doubly-exponential. Also, they need an intermediate game, the so-called block game, to prove the equivalence of the delay game and the parity game, while our equivalence proof is direct. Thus, our algorithm and its correctness proof are even simpler than the ones of Holtmann et al.

Second, we show that solving delay games is $\exptime$-complete by proving the first non-trivial lower bound on the complexity of $\omega$-regular delay games. The lower bound is proved by a reduction from the acceptance problem for alternating polynomial space Turing machines~\cite{ChandraKS81}, which results in delay games with safety conditions. Thus, solving delay games with safety conditions is already $\exptime$-hard. Our reduction is inspired by the $\exptime$-hardness proof for continuous simulation games~\cite{DBLP:journals/corr/HutagalungLL14}, a simulation game on Büchi automata where Duplicator is able to postpone her moves to obtain a lookahead on Spoiler's moves. However, this reduction is from a two-player tiling problem while we directly reduce from alternating Turing machines.

Third, we determine the exact amount of lookahead necessary to win delay games with $\omega$-regular winning conditions. From our algorithm we derive an exponential upper bound, which is again an exponential improvement. This upper bound is complemented by the first non-trivial lower bound on the necessary lookahead: there are reachability and safety conditions that are winning for the delaying player, but only with exponential lookahead, i.e., our upper bound is tight. 

Fourth, we present the first results for fragments of $\omega$-regular winning  conditions. As already mentioned above, our lower bounds on complexity and necessary lookahead already hold for safety conditions, i.e., safety is already as hard as parity. Thus, the complexity of the problems manifests itself in the transition structure of the automaton, not in the acceptance condition. For reachability conditions, the situation is different: we show that solving delay games with reachability conditions is equivalent to universality of non-deterministic reachability automata and therefore $\pspace$-complete.
Thus, there is a gap between the complexity of solving delay games with reachability conditions and delay games with safety conditions (unless $\pspace = \exptime$). Furthermore, as both reachability and safety already require exponential lookahead and are of high computational complexity it is natural to search for more tractable fragments. One such fragment is the class of winning conditions that are both reachability and safety, the so-called clopen languages. Here, we again prove tight bounds on the complexity and on the necessary lookahead for delay games with $\omega$-regular clopen winning conditions: linear lookahead is necessary and sufficient and determining the winner is $\conp$-complete.

Our results are summarized in Figure~\ref{fig_results}, where an asterisk~$*$ denotes cases where the lower bounds hold for deterministic automata while the upper bounds hold for non-deterministic automata. All bounds are tight.

\begin{figure}[h]

\begin{center}
\begin{tabular}{ c c c  }

acceptance & lookahead & complexity \\

\toprule

clopen$^*$ 					& linear 			& $\conp$-complete			\\
reachability$^*$ 			& exponential		& $\pspace$-complete		\\
det.\ safety 				& exponential 		& $\exptime$-complete		\\
det.\ parity					& exponential		& $\exptime$-complete		\\ 
\bottomrule

\end{tabular}
\end{center}

\caption{Table of results.}
\label{fig_results}	
\end{figure}


\section{Preliminaries}
\label{sec_defs}
The set of non-negative integers is denoted by $ \nats $. An
alphabet~$ \Sigma $ is a non-empty finite set of letters, $ \Sigma^{*}
$ the set of finite words over $ \Sigma $, $ \Sigma^{n} $ the set of
words of length~$ n $, and $ \Sigma^{\omega} $ the set of infinite words. The
empty word is denoted by $ \epsilon $ and the length of a finite word~$ w $
by~$ \size{w} $. For $ w \in \Sigma^{*} \cup \Sigma^{\omega} $ we
write $ \widx{w}{n} $ for the $ n $-th letter of $ w $.	

\subsection{Automata}
\label{subsec_automata}
An automaton~$ \aut = (Q, \Sigma, q_{ I }, \Delta, \varphi) $ consists
of a finite set~$ Q $ of states with initial state~$ q_{ I } \in Q $, an alphabet~$ \Sigma $, a non-deterministic transition function $ \Delta \colon Q \times \Sigma
\rightarrow 2^{Q} \setminus \set{ \emptyset } $, and an acceptance condition~$ \varphi $, which is either a set~$ F
\subseteq Q $ of accepting states, depicted by doubly-lined states, or
a coloring~$ \col \colon Q \rightarrow\nats $. Note that we require automata
to be complete, i.e., every state has at least one outgoing transition
labeled by $ a $ for every $ a \in \Sigma $. We discuss this
restriction after introducing the different acceptance conditions. The size of $
\aut $, denoted by $ \size{\aut} $, is the cardinality of $ Q $.  A
run of $ \aut $ on an infinite word~$ \widx{\alpha}{0} \widx{\alpha}{1} \widx{\alpha}{2} \cdots $ over $ \Sigma
$ is a sequence~$q_0 q_1 q_2 \cdots $ such that $q_0 =  q_{ I } $ and $q_{n+1} \in \Delta(q_n, \widx{\alpha}{n})$ for every $n$. Runs on finite words are defined analogously.

An automaton is deterministic, if $ |\Delta(q,a)| = 1 $ for every $ q
$ and $ a $. In this case, we denote $ \Delta $ by a function~$ \delta
\colon Q \times \Sigma \rightarrow Q $. A deterministic automaton has a unique
run on every finite or infinite word. A state~$ q $ of $ \aut $ is a sink, if $ \Delta(q,a) = \set{q}$ for every $ a
\in \Sigma $. 

Given an automaton~$ \aut $ over $ \Sigma $ with some set~$F$ of accepting states or with some coloring~$\col$, we consider the following
acceptance conditions and their induced languages:

\begin{itemize}

\item $ \lstar{\aut} \subseteq \Sigma^{*} $ denotes the set of finite
  words over $ \Sigma $ accepted by $ \aut $, i.e., the set of words
  that have a run ending in $ F $.

\item $ \lexists{\aut} \subseteq \Sigma^{\omega} $ denotes the set of
  infinite words over $ \Sigma $ that have a run visiting an accepting
  state at least once, called reachability acceptance. Due to completeness, we  have $ \lexists{\aut} = \lstar{\aut}
  \cdot \Sigma^{\omega} $.

\item Dually, $ \lall{\aut} \subseteq \Sigma^{\omega}$ denotes the set
  of infinite words over $ \Sigma $ that have a run only visiting
  accepting states, called safety acceptance.

\item $ \lpar{\aut} \subseteq \Sigma^{\omega} $ denotes the set of
  infinite words that have a run such that the maximal color visited
  infinitely often during this run is even. This is the classical (max)-parity
  acceptance condition.

\end{itemize}

\noindent A reachability (winning) condition is a language that is accepted by
an automaton with reachability acceptance, called a reachability
automaton. Safety (winning) conditions and parity (winning) conditions and safety and parity automata are defined
similarly. Every (deterministic) reachability automaton can be turned into a (deterministic) parity automaton of
the same size that recognizes the same language by turning the accepting
 states into sinks and by then defining
an appropriate coloring. Similarly, one can turn a (deterministic) safety automaton into a (deterministic) parity automaton of
the same size that recognizes the same language by turning the rejecting
 states into sinks. Deterministic parity
automata recognize exactly the $ \omega $-regular languages.

Recall that we require automata to be complete. For safety and parity
acceptance this is no restriction, since we can always add a fresh
rejecting sink, i.e., one that is not in $ F $ in the case of
safety and one with odd color in the case of parity, and lead all
missing transitions to this sink. However, incomplete automata with
reachability acceptance are strictly stronger than complete ones, as
incompleteness can be used to check safety properties. We impose this
restriction since we are interested in pure reachability conditions.

Given a language~$ L \subseteq \left( \SigmaI \times \SigmaO
\right)^{\omega} $ we denote by $ \proj{L} $ its projection to the
first component, i.e.,
$
\proj{L} = \set{ \widx{\alpha}{0} \widx{\alpha}{1} \widx{\alpha}{2}\cdots \mid
{\widx{\alpha}{0} \choose \widx{\beta}{0}}
{\widx{\alpha}{1} \choose \widx{\beta}{1}}
{\widx{\alpha}{2} \choose \widx{\beta}{2}} \cdots \in L}
$.
 Similarly, given an automaton~$ \aut $ over $ \SigmaI
\times \SigmaO$, we denote by $ \proj{\aut} $ the automaton obtained
by projecting each letter to its first component, i.e., we replace the alphabet~$\SigmaI \times \SigmaO$ by $\SigmaI$ and $\Delta$ by $\Delta'$, where $\Delta'(q,a) = \bigcup_{b \in \SigmaO}\Delta(q, {a \choose b})$ for all $ q \in Q $ and all $ a \in \SigmaI $.

\begin{rem}
  Let $ \acc \in \set{*, \exists, \forall, p}$, then $
  \proj{L_{\acc}(\aut)} = L_{\acc}(\proj{\aut}) $.
\end{rem}

\subsection{Games with Delay}
\label{subsec_delaygames}
A delay function is a mapping $ f \colon \nats \rightarrow\nats \setminus
\set{ 0 } $, which is said to be constant, if $ f(i) = 1 $ for
every $ i > 0 $. Given an $ \omega $-language $ L \subseteq
  \left( \SigmaI \times \SigmaO \right)^{\omega} $ and a delay
function~$ f $, the game $ \delaygame{L} $ is played by two players,
the input player \myquot{Player~$ I $} and the output player \myquot{Player~$ O
$} in rounds $ i = 0,1,2,\ldots $ as follows: in round~$ i $, Player~$ I $
picks a word $ u_{i} \in \SigmaI^{f(i)} $, then Player~$ O $ picks one
letter $ v_{i} \in \SigmaO $. We refer to the sequence $
(u_{0},v_{0}), (u_{1},v_{1}), (u_{2},v_{2}), \ldots$ as a play of $
\delaygame{L} $, which yields two infinite words $ \alpha =
u_{0}u_{1}u_{2} \cdots$ and $ \beta = v_{0}v_{1}v_{2} \cdots
$. Player~$ O $ wins the play if and only if the outcome~$ {
  \widx{\alpha}{0} \choose \widx{\beta}{0} }{ \widx{\alpha}{1} \choose
  \widx{\beta}{1} }{ \widx{\alpha}{2} \choose \widx{\beta}{2} } \cdots
$ is in $ L $, otherwise Player~$ I $ wins.

Given a delay function $ f $, a strategy for Player~$ I $ is a mapping
$ \stratI \colon \SigmaO^{*} \rightarrow\SigmaI^{*} $ where $
\size{\stratI(w)} = f(\size{w}) $, and a strategy for Player~$ O $ is
a mapping $ \stratO \colon \SigmaI^{*} \rightarrow\SigmaO $. Consider a
play $ (u_{0},v_{0}), (u_{1},v_{1}), (u_{2},v_{2}), \ldots $ of $
\delaygame{L} $. Such a play is consistent with $ \stratI $, if $
u_{i} = \stratI(v_{0} \cdots v_{i-1}) $ for every $ i \in \nats $. It
is consistent with $ \stratO $, if $ v_{i} = \stratO(u_{0} \cdots
u_{i}) $ for every $ i \in \nats $. A strategy~$ \strat $ for Player~$\p\in \set{I,O}$ is winning, if every play that is consistent with $\strat$ is winning for Player~$\p$. We say that a player wins $ \delaygame{L} $, if she has a winning strategy.

We continue with two examples of delay games with $ \omega
$-regular winning conditions.

\begin{exas} \label{example_introdelaygame}
  Note that both conditions can be accepted by safety automata.

  \begin{enumerate}

  \item \label{example_introdelaygame_predictnona}
  Consider $ L_{1} $ over $ \set{ a, b, c } \times \set{ b, c }
      $ with $ { \widx{\alpha}{0} \choose \widx{\beta}{0} }{
      \widx{\alpha}{1} \choose \widx{\beta}{1} }{ \widx{\alpha}{2}
      \choose \widx{\beta}{2} } \cdots \in L_{1} $, if $
    \widx{\alpha}{n} = a $ for every $ n \in \nats $ or if $
    \widx{\beta}{0} = \widx{\alpha}{n} $, where $ n $ is the smallest
    position with $ \widx{\alpha}{n} \neq a $. Intuitively, Player~$ O
    $ wins, if the letter she picks in the first round is equal to the
    first letter other than $ a $ that Player~$ I $
    picks. Also, Player~$ O $ wins, if there is no such
    letter.

    We claim that Player~$ I $ wins $ \delaygame{L_{1}} $ for every delay function~$
    f $: Player~$ I $ picks~$ a^{f(0)} $ in the first round and assume
    Player~$ O $ picks $ b $ afterwards (the case where she picks $ c
    $ is dual). Then, Player~$ I $ picks a word starting with $ c $ in
    the second round. The resulting play is winning for Player~$ I $
    no matter how it is continued. Thus, Player~$ I $ has a winning
    strategy in $ \delaygame{L_{1}} $.
 
  \item\label{example_introdelaygame_shift} Now, consider
    $ L_{2} $ over $ \set{ a, b, c} \times \set{ a, b, c} $ where $ { \widx{\alpha}{0} \choose \widx{\beta}{0} }{
        \widx{\alpha}{1} \choose \widx{\beta}{1} }{ \widx{\alpha}{2}
        \choose \widx{\beta}{2} } \cdots \in L_{2} $, if $
    \widx{\beta}{n} = \widx{\alpha}{n+2} $ for every $ n \in \nats $,
    i.e., Player~$ O $ wins if the input is shifted two positions to
    the left.

    Player~$ O $ has a winning strategy for $ \delaygame{L_{2}} $ for every $f$ with
    $ f(0) \geq 3 $. In this case, Player~$ O $ has at least three
    letters lookahead in each round, which suffices to shift the input
    of Player~$ I $ two positions to the left. On the other hand, if $
    f(0) < 3 $, then Player~$ I $ has a winning strategy, since
    Player~$ O $ has to pick $ \widx{\beta}{0} $ before $
    \widx{\alpha}{0+2} $ has been picked by Player~$ I $.

  \end{enumerate}
\end{exas}

\noindent Note that if a language~$ L $ is
recognizable by a (deterministic) parity automaton, then $ \delaygame{L} $ is
determined, i.e., exactly one of the players has a winning strategy, as a  delay
game with parity condition can be expressed as an explicit parity
game in a countable arena, which is determined~\cite{EmersonJutla91,Mostowski91}. This result has been recently generalized to the class of Borel winning conditions~\cite{KleinZimmermann15}.

Also, note that universality of $ \proj{L} $ is a necessary
condition for Player~$ O $ to win $ \delaygame{L} $. Otherwise,
Player~$ I $ could pick a word from $ \SigmaI^{\omega} \setminus
\proj{L} $, which is winning for him, no matter how Player~$ O $
responds.

\begin{prop} \label{prop_necccond}
  If Player~$ O $ wins $ \delaygame{L} $, then $ \proj{L} $ is
  universal.
\end{prop}

We are interested in solving delay games: given an automaton~$ \aut $
recognizing a
language~$ L \subseteq (\SigmaI \times \SigmaO)^{\omega} $, determine
whether Player~$ O $ wins $ \delaygame{L} $ for some delay
function~$ f $. Furthermore, we are interested in upper and lower
bounds on the lookahead induced by such an~$ f $. We measure the
complexity and the bounds in the size of the automation~$ \aut $.

%

\section{Lower Bounds on the Lookahead}
\label{sec_lowerbounds}
In this section, we prove lower bounds on the necessary lookahead for
Player~$ O $ to win delay games. We first give an exponential lower
bound for reachability conditions, then we extend this idea to provide
an exponential lower bound for safety conditions. Consequently, the same bounds
hold for more expressive acceptance conditions like Büchi, co-Büchi, and
parity. They are complemented by an exponential upper bound for parity
conditions in the next section. Note that both lower bounds already hold for deterministic automata.

\begin{thm} \label{thm_lowerboundsreach}
  For every $ n > 1 $ there is a language~$ L_{n} $ such that

  \begin{itemize}

  \item $ L_{n} = \lexists{\aut_{n}} $ for some deterministic
    automaton~$\aut_n$ with $ \size{\aut_{n}} \in \bigo(n) $,

  \item Player~$ O $ wins $ \delaygame{L_{n}} $ for some constant
    delay function~$ f $, but

  \item Player~$ I $ wins $ \delaygame{L_{n}} $ for every delay
    function~$ f $ with $ f(0) \leq 2^{n} $.

  \end{itemize}
\end{thm}

\begin{proof}
  Let $ \SigmaI = \SigmaO = \set{1, \ldots, n} $. We say that $ w $
  in $ \SigmaI^{*} $ contains a bad $ j $-pair, for $j \in \SigmaI $, if there are two occurrences of $ j $ in $ w $ such that no $ j'  > j $ occurs in between. The
  automaton~$ \autb_{j} $, depicted in
  Figure~\ref{fig_lowerboundreach}(a), accepts exactly the words with a bad $j$-pair. Now, consider the language $ L
  $ over $ \SigmaI $ defined by
  \begin{equation*}
    L = \bigcap\limits_{1 \leq j \leq n} \set{w \in \SigmaI^{*} \mid w
      \text{ contains no bad } j \text{-pair} } .
  \end{equation*}

\noindent First, we show that every $ w \in L $ satisfies $ \size{w} < 2^{n} $. To this end, we prove the stronger statement $ \size{w} < 2^{m} $, where $ m $ is the maximal letter occurring in $w$, by 
    induction over $m$. The induction base $ m = 1 $ is trivial, so let $ m > 1 $. There cannot be two occurrences of $m$ in $w$, as they would constitute a bad $m$-pair. Accordingly, there is exactly
    one $ m $ in $ w $, i.e., we can decompose $w$ into $ w =
    w_{\triangleleft} \, m \, w_{\triangleright} $ such that $w_{\triangleleft}$ and $ w_{\triangleright}$ contain no occurrence of $m$. Thus, the induction hypothesis is applicable and shows $ \size{w_{\triangleleft}}, \size{w_{\triangleright}} <  2^{m-1} $, which implies $ \size{w}
    < 2^{m} $.

   Dually, there is a word $ w_{n} \in L $ with $ \size{w_{n}} = 2^{n} -
    1 $ which is defined inductively via $ w_{1} = 1 $ and $ w_{m} =
    w_{m-1} \, m \; w_{m-1} $ for $ m > 1 $. A simple induction
    shows $ w_{n} \in L $ and $ \size{w_{n}} = 2^{n} -
    1 $.

  \begin{figure}[ht]
    \centering

    \begin{tikzpicture}
      \clip (-2.8,-1.5) rectangle (8.85,1.6);

      \node at (0,0.39) {
        \begin{tikzpicture}
          [thick, scale =.75, node distance = 2cm]
        
          \node[state, initial by arrow] (A) at (0,0) {};
          \node[state] (B) at (3,0) {};
          \node[state, accepting] (C) at (6,0) {};

          \path 
          
          (A) 
          edge[loop above] node[above] {$ \SigmaI \setminus \set{j}  $} (A)
          edge[bend left=20] node[above] {$ j $} (B)

          (B) 
          edge[loop above] node[above] {$ <\!j $} (B)
          edge[bend left=20] node[below] {$ >\!j $} (A)
          edge node[above] {$ j $} (C)

          (C) 
          edge[loop above] node[above] {$ \SigmaI $} (C)

          ; 
        \end{tikzpicture}
      };

      \node at (6.3,0) {
        \begin{tikzpicture}
          [thick, scale =0.95, node distance = 2cm]
	
          \node[state, initial by arrow] (q0p) at (0.5,0) {};
          \node[state, initial by arrow, initial distance=.25cm, 
          initial where = below, scale =.3] (q0) at (2, 1) {};
          \node[state, initial by arrow, initial distance=.25cm, 
          initial where = below, scale =.3] (qn) at (2, -1) {};

          \draw[rounded corners] (1.5, 0.5) rectangle (4.5, 1.5);
          \draw[rounded corners] (1.5, -.5) rectangle (4.5, -1.5);

          \node[anchor=east] at (4.5, 0.8) {$ \autb_{1}[a \backslash 
            {a \choose *}] $};
          \node[anchor=east] at (4.5, -1.2) {$ \autb_{n}[a \backslash 
            {a \choose *}] $};
          \node at (3.5, .1) {$\vdots$};

          \path 

          (q0p) 
          edge[bend left] 
          node[above, near start] {$ { * \choose 1 } $} (q0)
          edge[bend right]
          node[below, near start] {$ { * \choose n } $} (qn)

          ; 
        \end{tikzpicture}
      };

      \node at (-2.5,-1.2) {(a)};
      \node at (4.3,-1.2) {(b)};
    \end{tikzpicture}

    \caption{(a) Automaton~$ \autb_{j} $ for $ j \in \SigmaI
      $. (b) Construction of $ \aut_{n} $.}
    \label{fig_lowerboundreach}
  \end{figure}

\fussy
  \noindent The winning condition~$L_n$ is defined as follows: $ {
    \widx{\alpha}{0} \choose \widx{\beta}{0} }{ \widx{\alpha}{1}
    \choose \widx{\beta}{1} } { \widx{\alpha}{2} \choose
    \widx{\beta}{2} } \cdots $ is in $ L_{n} $ if $ \widx{\alpha}{1}
  \widx{\alpha}{2} \widx{\alpha}{3}\cdots $ contains a bad $ \widx{\beta}{0}
  $-pair, i.e., with her first move, Player~$O$ has to pick a $j$ such that Player~$I$ has produced a bad $j$-pair. For technical reasons, the first letter picked by Player~$I$ is ignored. The construction of an automaton~$ \aut_{n} $ recognizing $L_n$ is
  sketched in Figure~\ref{fig_lowerboundreach}(b), where $*$ denotes an arbitrary letter and $ \autb_{j}[a
  \backslash {a \choose *}] $ denotes $ \autb_{j} $ where for each $ a
  \in \SigmaI $ every transition labeled by $ a $ is replaced by
  transitions labeled by $ {a \choose b} $ for every $ b \in \SigmaO
  $. Clearly, we have $ \aut_{n} \in \bigo(n) $.
\sloppy
  
  Player~$ O $ wins $ \delaygame{L_{n}} $ for every delay function with $
  f(0) > 2^{n} $. In the first round, Player~$ I $ has to pick a word $ u_0$ such that $u_0$ without its first letter is not in $L$, as it is too long for being in $L$.
 This allows Player~$O$ to find a bad $j$-pair for some $j$, i.e., she wins the play no matter how it is continued.
 
 However, for $f$ with $ f(0) \leq 2^{n}, $ Player~$ I $
  has a winning strategy by picking the prefix of $ 1w_{n} $ of length $f(0)$ in the first round. Player~$ O
  $ has to answer with some $ j \in \SigmaO $. In this situation, Player~$ I $ continues by playing some $ j' \neq j $ ad infinitum, which ensures that the resulting sequence does not contain a bad $j$-pair. Thus, the play is winning for Player~$I$.
\end{proof}

 For safety conditions, we use the same idea as in the
reachability case, but we need to introduce a new letter $
\sharpsym $ to give Player~$ I $ the possibility to reach a non-accepting
state.

\begin{thm} \label{thm_lowerboundssafe}
  For every $ n > 1 $ there is a language~$ L_{n}' $ such that

  \begin{itemize}

  \item $ L_{n}' = \lall{\aut_{n}'} $ for some deterministic automaton~$\aut_n'$
    with $ \size{\aut_{n}'} \in \bigo(n) $,

  \item Player~$ O $ wins $ \delaygame{L_{n}'} $ for some constant
    delay function~$ f $, but

  \item Player~$ I $ wins $ \delaygame{L_{n}'} $ for every delay
    function~$ f $ with $ f(0) \leq 2^{n} $.

  \end{itemize}
\end{thm}

\begin{proof}
  Let $ \Sigma_{I} = \Sigma_{O} = \set{1,2,\ldots,n,\sharpsym} $ and let $
  w_{n} $ be defined as above. We introduce a new automaton $
  \autbsharp' $ and extend every automaton $ \autb_{j} $  from the previous proof to $
  \autb_{j}' $ as depicted in Figure~\ref{fig_lowerboundsafety}. The
  automaton $ \aut_{n}' $ is constructed as in the previous proof
  using the automata~$ \autb_{j}' $ and $\autb_{\sharpsym}'$ instead of the $ \autb_{j} $.
  
  \begin{figure}[ht]
    \centering

    \begin{tikzpicture}
      \clip (-2.75,-1.99) rectangle (9.3,1.75);

      \node at (0,0) {
        \begin{tikzpicture}
          [thick, scale =.7, node distance = 2cm]
      
          \node[state, accepting, initial by arrow] (A) at (0,0) {};
          \node[state, accepting] (B) at (3,0) {};
          \node[state, accepting] (C) at (6,0) {};
          \node[state] (D) at (1.5,-2.2) {};

          \path 

          (A) 
          edge[loop above] node[above] {$ \SigmaI \setminus \set{j,\sharpsym} $} (A)
          edge[bend left=20] node[above] {$ j $} (B)
          edge node[left] {$ \sharpsym $} (D)

          (B) 
          edge[loop above] node[above] {$ <\!j $} (B)
          edge[bend left=20] node[below] {$ >\!j $} (A)
          edge node[above] {$ j $} (C)
          edge node[right] {$ \sharpsym $} (D)

          (C) 
          edge[loop above] node[above] {$ \SigmaI $} (C)

          (D) 
          edge[loop right] node[right] {$ \SigmaI $} (D)

          ; 
        \end{tikzpicture}
      };

      \node at (6.4,0) {
        \begin{tikzpicture}
          [thick, scale =.7, node distance = 2cm]
      
          \node[state, accepting, initial by arrow] (A) at (0,0) {};
          \node[state, accepting] (B) at (3,0) {};
          \node[state] (C) at (6,0) {};
          \node[state, accepting] (D) at (1.5,-2.2) {};

          \path 

          (A) 
          edge[loop above] node[above] {$ \SigmaI \setminus \set{n,\sharpsym} $} (A)
          edge[bend left=20] node[above] {$ n $} (B)
          edge node[left] {$ \sharpsym $} (D)

          (B) 
          edge[loop above] node[above] {$ <\!n $} (B)
          edge node[above] {$ n $} (C)
          edge node[right] {$ \sharpsym $} (D)

          (C) 
          edge[loop above] node[above] {$ \SigmaI $} (C)

          (D) 
          edge[loop right] node[right] {$ \SigmaI $} (D)

          ; 
        \end{tikzpicture}
      };

      \node at (-2.3,-1.6) {(a)};
      \node at (4.1,-1.6) {(b)};
    \end{tikzpicture}

    \caption{(a) Automaton~$ \autb_{j}' $ for $ j \in \SigmaI \setminus
      \set{ \sharpsym } $. (b) Automaton~$ \autbsharp' $. }
    \label{fig_lowerboundsafety}
  \end{figure}

 \noindent For $
  f(0) > 2^{n} $, Player~$ O $ wins the game: assume Player~$ I $
  picks $ u_0 $ in the first round and let $u_0'$ be $u_0$ without its first letter. If $u_0'$ contains a~$ \sharpsym $
  preceded by at most one $ n $, then Player~$ O $ answers with $
  \sharpsym $ in the first round. If there is more than one $ n $ before the first $\sharpsym$ in $u_0'$, then she answers with $ n
  $. Finally, if there is no $ \sharpsym $ in $ u_0' $, she can pick a $ j $ such that $u_0'$ contains a bad $j$-pair. All outcomes are winning for Player~$O$.
  
  Player~$ I $ still wins the game for a constant delay function~$f$ with $ f(0) \leq
  2^{n} $ by picking the prefix of $1 w_{n} $ of length~$f(0)$ in the first round: if Player~$ O $ picks some $ j
  \in \Sigma_{O} \setminus \set{\sharpsym} $ in the first round, then Player~$ I $ just
  has to answer with $ \sharpsym $. Otherwise, if Player~$ O $ picks $
  \sharpsym $ in the first round, then Player~$ I $ continues with~$ n^{\omega} $. He wins in both situations. 
\end{proof}

The aforementioned constructions also work for constant-size
alphabets, if we encode every $ j \in \set{1,\ldots,n} $ in binary
with the most significant bit in the first position. Then, the natural ordering
on $ \set{ 1, \ldots, n } $ is exactly the lexicographical ordering on
the corresponding bit-string representation. Accordingly, we can
encode every $ \autb_{j} $, $ \autb_{j}' $, and $ \autbsharp' $ in
logarithmic size in $ n $, as deciding whether the input
represents $ j $, is larger than $ j $, or smaller than $ j $ can be
checked bit-wise. Together with a binary decision tree of size $\bigo(n)$ for the initial choice
of Player~$ O $ we obtain deterministic automata~$ \aut_{n} $ and $ \aut_{n}' $ whose sizes are in $ \bigo(n\log n) $. It is open whether linear-sized automata and a constant-sized alphabet can be achieved simultaneously.


\section{Computational Complexity of Delay Games}
\label{sec_complexity}
In this section, we determine the computational complexity of solving
delay games. First, we consider the special case of reachability
conditions and prove such games to be $ \pspace $-complete. Then, we
show that games with safety conditions are $ \exptime $-hard. The latter
 bound is complemented by an $ \exptime $-algorithm for solving
delay games with parity conditions. From this algorithm, we also
deduce an exponential upper bound on the necessary lookahead for
Player~$ O $, which matches the lower bounds given in the previous
section.

\subsection{Reachability Conditions}
\label{subsec_reachcomplex}
Recall that universality of the projection to the first component of the winning condition is
a necessary condition for Player~$ O $ for having a winning strategy in a
delay game. Our first result in this section states that universality
is also sufficient in the case of reachability winning conditions. Thus, solving delay games with reachability conditions is equivalent, via linear time reductions, to the universality problem for non-deterministic reachability automata, which is $\pspace$-complete (see Appendix~\ref{sec_universality}). Therefore, solving delay games with reachability conditions is $\pspace$-complete as well. Also, our proof yields an exponential upper bound on the necessary lookahead.

\begin{thm} \label{theorem_reachcomplex} Let $ L = \lexists{\aut} $,
  where $ \aut $ is a non-deterministic reachability automaton. The following are equivalent:

  \begin{enumerate}
  
  \item \label{reach_Owins} Player~$ O $ wins $ \delaygame{L} $ for
    some delay function~$ f $.

  \item \label{reach_Owinswithconstantdelay} Player~$ O $ wins $
    \delaygame{L} $ for some constant delay function~$ f $ with $ f(0)
    \leq 2^{\size{\aut}} $.

  \item\label{reach_universality} $ \proj{L} $ is universal.

  \end{enumerate}
\end{thm}

\begin{proof} 
  The implication~(\ref{reach_Owinswithconstantdelay})~$ \Rightarrow $~(\ref{reach_Owins}) is trivial and the implication~(\ref{reach_Owins})~$ \Rightarrow
  $~(\ref{reach_universality}) is given by
  Proposition~\ref{prop_necccond}.  It remains to show
  (\ref{reach_universality})~$ \Rightarrow $~(\ref{reach_Owinswithconstantdelay}). 

 Let $\proj{L}$ be universal. First, recall that the projection automaton $\proj{\aut}$ recognizes $\proj{L}$ and that $\lexists{\proj{\aut}} = \lstar{\proj{\aut}} \cdot \SigmaI^\omega$. Using a pumping argument, we show that on every input $\proj{\aut}$ has a run reaching an accepting state after at most $2^{\size{\aut}}$ steps. Thus, given the first move~$\alpha(0) \cdots \alpha(f(0)-1)$ of Player~$I$ with $f(0) \ge 2^{\size{A}}$, Player~$O$ wins by picking a suitable completion that forces the run of $\aut$ into an accepting state.
 
  To formalize this, we assume w.l.o.g.\ that the
  accepting states of $ \aut $ are sinks, which implies that $
  \lstar{\proj{\aut}} $ is suffix-closed, i.e., $ w \in
  \lstar{\proj{\aut}} $ implies $ ww' \in \lstar{\proj{\aut}} $ for
  every $ w' \in \SigmaI^{*} $. Furthermore, let $ \aut^{c} $ be
  an automaton recognizing the complement of $ \lstar{\proj{\aut}} $,
  which is prefix-closed, as it is the complement of a suffix-closed
  language. We can choose $ \aut^{c} $ such that $ \size{\aut^{c}}
  \leq 2^{\size{\aut}} $.

  We claim that $ \lstar{\aut^{c}} $ is finite. Assume it is infinite.
  Then, by König's Lemma there is an infinite word~$ \alpha $ whose
  prefixes are all in $ \lstar{\aut^{c}} $. Due to universality, we
  have $ \alpha \in \lexists{\proj{\aut}} $, i.e., there is a
  prefix of $ \alpha $ in $ \lstar{\proj{\aut}} $. Thus, the prefix is
  in $\lstar{\proj{\aut}}$ and in the complement~$ \lstar{\aut^c} $
  yielding the desired contradiction. An automaton with $ n $ states
  with a finite language accepts words of length at
  most~$ n-1 $. Thus, $ w \in \lstar{\proj{\aut}} $ for
  every $ w \in \SigmaI^{*} $ with $ |w| \geq 2^{\size{\aut}} $.

  Using this, we show that Player~$ O $ wins $ \delaygame{L} $ if $
  f(0) \ge 2^{\size{\aut}} $. Player~$ I $ has to pick $ f(0) $ letters
  with his first move, say $ u_{0} = \widx{\alpha}{0} \cdots
  \widx{\alpha}{f(0)-1}$. As $ f(0) $ is large enough, we have $
  u_{0} \in \lstar{\proj{\aut}} $. Hence, there is a word~$
  \widx{\beta}{0} \cdots \widx{\beta}{f(0)-1} \in \SigmaO^{*} $ such
  that ${\alpha(0) \choose \beta(0)} \cdots {\alpha(f(0)-1) \choose
    \beta(f(0)-1)} \in \lstar{\aut}$. By picking $
  \widx{\beta}{0}, \ldots, \widx{\beta}{f(0)-1} $ in the first $ f(0)
  $ rounds, Player~$ O $ wins the play, no matter how it is
  continued. Hence, she has a winning strategy.
\end{proof}

\noindent The exponential upper bound on the necessary lookahead to
win delay games with reachability conditions matches the lower bound
presented in the previous section. 

Theorem~\ref{theorem_reachcomplex} shows that
solving delay games with reachability conditions is equivalent to
universality of non-deterministic reachability automata, which is $\pspace$-complete (see Appendix~\ref{sec_universality}). Thus, we obtain the complexity of solving delay games with reachability conditions as a corollary of Theorem~\ref{theorem_reachcomplex}.

\begin{cor}
  The following problem is $ \pspace $-complete: Given a non\--deter\-ministic reachability
  automaton~$ \aut $, does Player~$ O $ win
  $ \delaygame{\lexists{\aut}} $ for some $ f $?
\end{cor}

Another consequence of the proof of
Theorem~\ref{theorem_reachcomplex} concerns the strategy complexity of
delay games with reachability conditions: if Player~$O$ wins for some
delay function, then she has a winning strategy that receives
exponentially many input letters and answers by also giving
exponentially many output letters and thereby already guarantees a
winning play, i.e., all later moves are irrelevant. Thus, the
situation is similar to classical reachability games on graphs, in
which positional attractor strategies allow a player to guarantee a
win after a bounded number of moves. The strategy described above can
be implemented by a lookup table that maps all minimal words in $
\lstar{\proj{\aut}} $ to a word in $ \SigmaO^{*} $ of the same length
such that the combined word is accepted by $ \aut $.

Finally, note that both upper bounds, the one on the lookahead and the one on the complexity, hold for non-deterministic automata while the lower bounds already hold for deterministic automata.

\subsection{Safety Conditions}
\label{subsec_safetycomplex}
Unsurprisingly, Example~\ref{example_introdelaygame}.\ref{example_introdelaygame_predictnona}
shows that Theorem~\ref{theorem_reachcomplex} does not hold for safety
conditions: the projection~$ \proj{L_{1}} $ is universal, but Player~$
O $ has no winning strategy for any delay function. It turns out that safety conditions are even harder than reachability conditions (unless $\pspace$ equals $\exptime$): we show solving delay games with safety conditions to be $ \exptime$-hard by a reduction from the acceptance problem for alternating polynomial space Turing machines~\cite{ChandraKS81}. 

\begin{thm}
	\label{thm_safetyhardness}
  The following problem is $ \exptime $-hard: Given a deterministic
  safety automaton $ \aut $, does Player~$ O $ win $
  \delaygame{\lall{\aut}} $ for some $ f $?
\end{thm}

\begin{proof}
Let $ \tm = (Q, Q_\exists, Q_\forall, \Sigma, q_I, \Delta, \qacc, \qrej)$ be an alternating polynomial space Turing machine, where $\Delta \subseteq Q \times \Sigma \times Q \times \Sigma \times \set{-1, 0, 1}$ is the transition relation, and let $x \in \Sigma^*$ be an input. For technical reasons, we assume the accepting state~$\qacc$ and the rejecting state~$\qrej$ to be equipped with a self-loop. Furthermore, let $p$ be a polynomial that bounds $\tm$'s space consumption. We construct a safety automaton~$\aut$ of polynomial size in $\size{\Delta}$ and $p(\size{x})$ such that $\tm$ rejects $x$ if and only if Player~$O$ wins $\delaygame{\lall{\aut}}$ for some $f$. This is sufficient, since $\apspace = \exptime$~\cite{ChandraKS81} is closed under complement. Thus, we give Player~$I$ control over the existential states while Player~$O$ controls the universal ones. Additionally, Player~$I$ is in charge of producing all configurations with his moves. He can copy configurations in order to wait for Player~$O$'s choices for the universal transitions, which are delayed due to the lookahead.

Formally, the input alphabet~$\SigmaI$ contains $\Sigma \cup Q$ and two separators~$\fresh$ and $\ccopy$ while the output alphabet~$\SigmaO$ contains $\Delta$ and two signals $\copyerror$ and $\noerror$. Intuitively, Player~$I$ produces configurations of $\tm$ of length $p(\size{x})$ preceded by either $\ccopy$ or $\fresh$ to denote whether the configuration is a \emph{copy} of the previous one or a \emph{new} one. Copying configurations is necessary to bridge the lookahead while waiting for Player~$O$ to determine the transition that is applied to a universal configuration. Player~$I$ could copy a configuration ad infinitum, but this will be losing for him, unless it is an accepting one. Player~$O$ chooses universal transitions at every separator\footnote{If the following configuration is existential or the separator is a $\ccopy$, then her choice is ignored.} by picking a letter from $\Delta$. At every other position, she has to pick a signal: $\copyerror$ allows her to claim an error in the configurations picked by Player~$I$ while $\noerror$ means that she does not claim an error at the current position.

The automaton~$\aut$ is the product of safety automata checking the following properties of an input word~${\widx{\alpha}{0} \choose \widx{\beta}{0}} {\widx{\alpha}{1} \choose \widx{\beta}{1}} {\widx{\alpha}{2} \choose \widx{\beta}{2}} \cdots \in (\SigmaI \times \SigmaO)^\omega$:
\begin{enumerate}
	\item $\alpha \in (\set{\fresh, \ccopy} \cdot \conf)^\omega$, where $\conf$ is the set of encodings of configurations of length~$p(\size{x})$, i.e., words of length $p(\size{x})+1$ over $\Sigma \cup Q$ that contain exactly one letter from $Q$. If this is not the case, then the product automaton~$\aut$ goes to an accepting sink, i.e., in order to win, Player~$I$ has to produce an $\alpha$ that satisfies the requirement.

\item $\beta \in (\Delta \cdot \set{\noerror, \copyerror}^{p(\size{x})+1})^\omega$. If this is not the case, then the product automaton~$\aut$ goes to a rejecting sink, i.e., in order to win, Player~$O$ has to produce a $\beta$ that satisfies the requirement.

\item The first configuration picked by Player~$I$ is the initial one of $\tm$ on $x$. If this is not the case, the product automaton~$\aut$ goes to an accepting sink.

\item If $\beta$ contains a $\copyerror$, then the automaton checks whether there is indeed an error by doing the following at the first occurrence of $\copyerror$: it stores the previous, the current, and the next input letter, the transition picked by Player~$O$ at the last separator~$\fresh$, and whether the current configuration is existential or universal. Some of this information has to be stored continuously, since these letters appear before the first $\copyerror$. This is possible using a set of states whose size is polynomial in $\size{\Sigma} +  \size{ Q } + \size{\Delta}$. 

Then, the automaton processes $p(\size{x})+1$ letters (and remembers whether it traverses the separator~$\fresh$ or $\ccopy$), and then checks whether the letter just reached is updated correctly or not:
\begin{itemize}
	\item If the separator is $\ccopy$, then the current letter is updated correctly, if it is equal to the marked one.
	\item If the separator is $\fresh$ and the configuration in which the error was marked is existential, then the letter is updated correctly, if there is a transition of $\tm$ that is compatible with the current letter and the marked one.
	\item If the separator is $\fresh$ and the configuration in which the error was marked is universal, then the letter is updated correctly, if it is compatible with the transition picked by Player~$O$ at the last separator~$ N $ before the $\copyerror$, which is stored by the automaton. If she has picked a transition that is not applicable to the current configuration, the product automaton~$\aut$ goes to a rejecting sink.

\end{itemize}
If the update is not correct, i.e., Player~$O$ has correctly claimed an error, then $\aut$ goes to an accepting sink. Otherwise, it goes to a rejecting sink, i.e., in order to win, Player~$O$ should only claim an error at an incorrect update of a configuration, but she wins if she correctly claims an error. All subsequent claims by Player~$O$ are ignored, i.e., after the first claim is evaluated, the play is either accepted or rejected, no matter how it is continued. 

\item Finally, if $\alpha$ contains the accepting state of $\tm$, then $\aut$ goes to a rejecting sink, unless Player~$O$ correctly claimed an error in a preceding configuration.

\end{enumerate}
All these properties can be checked by deterministic safety automata whose sizes are polynomial in the size of $\tm$ and $p(\size{x})$. All non-sink states of $\aut$ are accepting, i.e., as long as both players stick to their requirements on the format, Player~$I$ starts with the initial configuration, Player~$O$ does not incorrectly claim an error, and the accepting state of $\tm$ is not reached, then the input is accepted.

It remains to prove that $\tm$ rejects $x$ if and only if Player~$O$ wins $\delaygame{\lall{\aut}}$ for some~$f$.

\myquot{$\Rightarrow$}: Assume $\tm$ rejects $x$ and let $f$ be the constant delay function with $f(0) =  p(\size{x}) +3$. We show that Player~$O$ wins $\delaygame{\lall{\aut}}$. At every time, Player~$O$ has enough lookahead to correctly claim the first error introduced by Player~$I$, if he introduces one. Furthermore, she has access to the whole encoding of each universal configuration whose successor she has to determine. This allows her to simulate the rejecting run of $\tm$ on $x$, which does not reach the accepting state~$\qacc$, no matter which transitions Player~$I$ picks. Thus, he has to introduce an error in order to win, which Player~$O$ can detect using the lookahead. If Player~$I$ does not introduce an error, the play proceeds ad infinitum by repeating a rejecting configuration forever. In every case, $\aut$ accepts the resulting play, i.e., Player~$O$ wins. Thus, Player~$O$ has a winning strategy for $\delaygame{\lall{\aut}}$.

\myquot{$\Leftarrow$}: We show the contrapositive. Assume that $\tm$ accepts $x$ and let $f$ be an arbitrary delay function. We show that Player~$I$ wins $\delaygame{\lall{\aut}}$. Player~$I$ starts with the initial configuration and picks the successor configuration of an existential one according to the accepting run, and copies universal configurations as often as necessary to obtain a play prefix in which Player~$O$ has to determine the transition she wants to apply in this configuration. Thus, he will eventually produce an accepting configuration of $\tm$ without ever introducing an error. Hence, either Player~$O$ incorrectly claims an error or the play reaches an accepting state. In either case, Player~$I$ wins the resulting play, i.e., he has a winning strategy for $\delaygame{\lall{\aut}}$.
\end{proof}

It is noteworthy that the lower bound just proven does not require the full exponential lookahead that is in general necessary to win delay games with safety conditions: Player~$O$ wins the game constructed above with sublinear lookahead, as $p(\size{x})+3$ is smaller than the size of $\aut$. Thus, determining the winner of a delay game with safety condition with respect to linearly bounded delay is already $\exptime$-hard.

Finally, the lower bound just proven already holds for deterministic safety automata.

\subsection{Parity Conditions}
\label{subsec_paritycomplex}
In the previous two subsections, we showed solving delay games with
reachability conditions to be $ \pspace $-complete and solving games with safety conditions to be $\exptime$-hard. To conclude this section, we
complement the latter with an exponential time algorithm for
solving delay games with parity conditions. Thus, delay games with safety or parity conditions are $\exptime$-complete.
Also, we derive an exponential upper bound of the form~$2^{(nk)^2}$ on the necessary lookahead from the algorithm, where $n$ is the size and $k$ the number of colors of the automaton. Finally, we lower the upper bound to $2^{nk}$ via a direct pumping argument, which matches the lower bound from the previous section. Note that all results only hold for deterministic
automata.

\begin{thm} \label{thm_paritycomplexity_ub} The following problem is
  in $ \exptime $: Given a deterministic parity automaton~$ \aut $, does
  Player~$ O $ win $ \delaygame{\lpar{\aut}} $ for some delay function
  $ f $?
\end{thm}

We proceed by constructing an exponentially-sized,
delay-free parity game with the same number of colors as $\aut$, which is won by
Player~$ O $ if and only if she wins $ \delaygame{\lpar{\aut}} $ for
some delay function $ f $. Intuitively, we assign to each potential lookahead~$w \in \SigmaI^*$ the behavior it induces on $\aut$, which is given by a function~$\resolve \colon Q \rightarrow 2^{Q \times \col(Q)}$. If $(q', c) \in \resolve(q)$, then $\aut$ has a run from $q$ to $q'$ with maximal color~$c$ on a word over $\SigmaI \times \SigmaO$ whose projection to $\SigmaI$ is $w$. Having the same behavior gives rise to an equivalence relation over $\SigmaI^*$ of exponential index. In the parity game we construct, Player~$I$ picks equivalence classes of this relation and Player~$O$ constructs a run on representatives. Player~$O$ wins, if this run is accepting. To account for the delay in the original game, Player~$I$ is always one move ahead. This gives Player~$O$ a lookahead of one equivalence class, which can be stored in the state space of the parity game. 

First, we adapt $ \aut $ to keep track of the maximal color visited during a run. Let $ \aut = (Q, \SigmaI \times \SigmaO, q_{I}, \delta, \col )$ with $
\col \colon Q \rightarrow\nats $ and let $ n = \size{Q \times \col(Q)} $. We define the color-tracking automaton~$ \autc = (Q_{ \autc }, \SigmaI \times \SigmaO, q_{I}^{ \autc },
\delta_{ \autc }, \col_{ \autc }) $  of size $n$ where 
\begin{itemize}
	\item $Q_{ \autc } = Q \times \col(Q) $,
	\item $ q_{I}^{ \autc } = (q_{I}, \col(q_{I})) $,
	\item $ \delta_{ \autc }((q,c), a) = (\delta(q,a), \max\set{ c, \col(\delta(q,a)) }) $, and
	\item $ \col_{ \autc }(q,c) = c $.
\end{itemize}
Note that $\autc $ does not recognize $ \lpar{\aut} $. However, we are not
interested in complete runs of $ \autc $, but only in runs on finite
play infixes.

\begin{rem}
\label{remark_autproduct}	
Let $w \in (\SigmaI \times \SigmaO)^*$ and let $(q_0, c_0)(q_1, c_1) \cdots (q_{|w|}, c_{|w|})$ be the run of $\autc$ on $w$ from some state~$(q_0, c_0) \in \set{ (q, \col(q)) \mid q \in Q}$. Then, $q_0 q_1 \cdots q_{|w|}$ is the run of $\aut$ on $w$ starting in $q_0$ and $c_{|w|} = \max\set{\col(q_{j}) \mid 0 \leq j \leq \size{w} } $. 
\end{rem}

In the following, we work with partial functions from $ Q_{ \autc } $ to $ 2^{Q_{ \autc }}
$, where we denote the domain of each such function~$ r $ by $ \dom(r)
$. Intuitively, we use $ r $ to capture the information encoded in the
lookahead provided by Player~$ I $. Assume Player~$ I $ has picked $
\widx{\alpha}{0} \cdots \widx{\alpha}{j} $ and
Player~$ O $ has picked $ \widx{\beta}{0} \cdots
\widx{\beta}{i} $ for $ i < j $ such that the lookahead is $ w =
\widx{\alpha}{i+1} \cdots \widx{\alpha}{j} $. Then,
we can determine the state~$ q $ that $ \autc $ reaches after
processing $ { \widx{\alpha}{0} \choose \widx{\beta}{0} }
  \cdots { \widx{\alpha}{i}
  \choose \widx{\beta}{i} } $, but the automaton cannot process $ w $,
since Player~$ O $ has not yet picked $ \widx{\beta}{i+1}
\cdots \widx{\beta}{j} $. However, we can determine
the states Player~$ O $ can enforce by picking an appropriate
completion, which will be the ones contained in $ r(q) $.

 To formalize the functions capturing the lookahead picked by Player~$I$, we
define $ \delta_{\autp} \colon 2^{Q_{ \autc }} \times \SigmaI \rightarrow2^{Q_{ \autc }} $ via $ \delta_{\autp}(S, a) =
  \bigcup_{q \in S} \bigcup_{b \in \SigmaO} \delta_{ \autc }(q,{ a \choose b })$,
i.e., $ \delta_{\autp} $ is the transition function of the powerset
automaton of $ \proj{\autc} $. As usual, we extend $ \delta_{\autp} $ to $
\delta_{\autp}^{*} \colon 2^{Q_{ \autc }} \times \SigmaI^{*} \rightarrow2^{Q_{ \autc }} $ via $
\delta_{\autp}^{*}(S, \epsilon) = S $ and $ \delta_{\autp}^{*}(S, wa) =
\delta_{\autp}(\delta_{\autp}^{*}(S,w), a) $.

Let $D \subseteq Q_{ \autc }$ be non-empty and let $w \in \SigmaI^*$. We define the function~$\resolve_w^D$ with domain $D$ as follows: for every $(q,c) \in D$, we have 
\[ \resolve_w^D(q,c) = \delta_{\autp}^*(\set{(q,\col(q))},w). \]
Note that we apply $\delta_{\autp}^{*}$ to $\set{(q, \col(q))}$, i.e., the second component is the color of $q$ and not the color~$c$ from the argument to $\resolve_w^D$. 
If $(q', c') \in \resolve_w^D(q,c)$, then there is a word~$w'$ whose projection is $w$ and such that the run of $\aut$ on $w'$ leads from $q$ to $q'$ and has maximal color~$c'$. Thus, if Player~$I$ has picked the lookahead~$w$, then Player~$O$ could pick an answer such that the combined word leads $\aut$ from $q$ to $q'$ with maximal color~$c'$. 
 
We call $w$ a witness for a partial function $\resolve \colon Q_{ \autc } \rightarrow 2^{Q_{ \autc }}$, if we have $\resolve = \resolve_w^{\dom(r)}$. Thus, we obtain a language~$\wit{\resolve} \subseteq \SigmaI^*$ of witnesses for each such function~$\resolve$. We define 
\[\curlyR = \set{\resolve \mid \dom(\resolve)\neq \emptyset \text{ and } \wit{\resolve} \text{ is infinite}}.\]

\begin{lem}
\label{lemma_graphautprops}
Let $\curlyR$ be defined as above.
\begin{enumerate}
	
	\item\label{lemma_graphautprops_witnonempty}
	Let $\resolve \in \curlyR$. Then, $\resolve(q) \not= \emptyset$ for every $q \in \dom(\resolve)$.
	
	\item\label{lemma_graphautprops_witdisjoint}
	Let $\resolve \neq \resolve' \in \curlyR$ such that  $\dom(\resolve) = \dom(\resolve')$. Then, $\wit{\resolve} \cap \wit{\resolve'} = \emptyset$.

	\item 
	\label{lemma_graphautprops_aut}
	Let $\resolve$ be a partial function from $Q_\autc$ to $2^{Q_\autc}$ with non-empty domain. Then, $\wit{\resolve}$ is recognized by a deterministic finite automaton with at most $2^{n^2}$ states.
			
	\item\label{lemma_graphautprops_witcomplete}
	Let $D \subseteq Q_{ \autc }$ be non-empty and let $w \in \SigmaI^*$ be such that $|w| \ge 2^{n^2}$. Then, there exists some $\resolve \in \curlyR$ with $\dom(\resolve) = D$ and $w \in \wit{\resolve}$.
	
\end{enumerate}
\end{lem}

\begin{proof} The first statement follows from completeness of the automata~$\aut$, $\autc$, and $\autp$ while the second one follows from the definition of $\resolve_w^D$, which is uniquely determined by $w$ and $D$. Hence, a fixed $w$ cannot witness two different functions~$\resolve$ and $\resolve'$ with the same domain. 

To prove the third statement, fix some  partial function $\resolve $ from $Q_\autc$ to $2^{Q_\autc}$ with domain~$D = \set{(q_1, c_1), \ldots, (q_{\size{D}}, c_{\size{D}})}$. Then, the product of $\size{D}$ copies of the automaton~$\autp$ with the initial state~$( \set{(q_1, \col(q_1))}, \ldots, \set{(q_{\size{D}}, \col(q_{\size{D}}))} )$ and the unique accepting state~$(\resolve(q_1,c_1), \ldots, \resolve(q_{\size{D}}, c_{\size{D}}))$ recognizes the witness language~$W_r$. As $\size{D} \le n$, the automaton has at most $2^{n^2}$ states. 

For proving the last statement, we fix some non-empty $D$ and some $w$ of length at least $2^{n^2}$. Define $\resolve = \resolve_w^D$, which implies $w \in \wit{\resolve}$ by definition. As just shown, there exists an automaton recognizing $\wit{\resolve}$ with at most $2^{n^2} \le \size{w}$ many states. Thus, the accepting run of the automaton on $w$ contains a state-repetition. Hence, $\wit{\resolve}$ is infinite, i.e., $\resolve \in \curlyR$.
\end{proof}

Now, we are able to define the equivalent delay-free parity game. As already alluded to, Player~$I$ picks elements from $\curlyR$ while Player~$O$ produces a run on witnesses, which corresponds to picking suitable completions to witnesses of the functions picked by Player~$I$. She wins, if the constructed run is accepting. Finally, to account for the lookahead, Player~$I$ is always one move ahead. 

To keep the equivalence proof between the two games simple,  we first give an abstract description of the delay-free game. Then, in the proof of our main theorem, we show how to model this game as a classical graph based parity game, which is solvable in exponential time.

Formally, the game~$ \game(\aut) $ 
is played between Player~$ I $ and Player~$ O $ in rounds $ i = 0, 1,
2, \ldots $ as follows: in each round, Player~$ I $ picks a function
from $ \curlyR $ and Player~$ O $ answers by a state of $\autc$ subject to the following constraints. In the first round, Player~$ I $ has to pick $ r_{0} \in \curlyR$ such that
\begin{equation} \label{constraint1}\tag{C1}
  \dom(r_{0}) = \set{q_{I}^{ \autc }}
\end{equation}
and Player~$ O $ has to answer by picking a state $ q_{0} \in \dom(r_0)$, which implies $q_0 = q_{I}^{ \autc } $.
Now, consider round $ i > 0 $: Player~$I$ has picked functions~$
r_{0}, r_{1}, \ldots, r_{i-1} $ and Player~$ O $ has picked states~$ q_{0},
q_{1}, \ldots, q_{i-1} $ with $ q_{i-1} \in \dom(\resolve_{i-1}) $. Next, Player~$ I $ has to pick a
function~$ r_{i}  \in \curlyR$ such that 
\begin{equation} \label{constraint2}\tag{C2}
  \dom(r_{i}) = r_{i-1}(q_{i-1}) .
\end{equation}
Afterwards, Player~$ O $ picks some state~$ q_{i} \in \dom(r_{i}) $. 

Both
players can always move: Player~$ I $ can move, as $ r_{i-1}(q_{i-1}) $ is
always non-empty (Lemma~\ref{lemma_graphautprops}.\ref{lemma_graphautprops_witnonempty}) and thus the domain of some $\resolve \in \curlyR$ (Lemma~\ref{lemma_graphautprops}.\ref{lemma_graphautprops_witcomplete}), Player~$ O $ can move, as the domain of every~$ r \in
\curlyR $ is non-empty by construction. The resulting play of $
\game(\aut) $ is the sequence~$ r_{0} q_{0} r_{1} q_{1} r_{2} q_{2}
\cdots $, which is won by Player~$O$ if the maximal color occurring infinitely
often in $ \col_{ \autc }(q_{0})\col_{ \autc }(q_{1})\col_{ \autc }(q_{2}) \cdots$ is
even. Otherwise, Player~$ I $ wins.

A strategy for Player~$ I $ is a function~$ \stratI' $ mapping the
empty play prefix to a function~$ r_{0} $ satisfying
(\ref{constraint1}) and mapping a non-empty prefix~$ r_{0} q_{0}
\cdots r_{i-1} q_{i-1} $ to a function~$ \resolve_{i} $ satisfying
(\ref{constraint2}). A strategy~$\stratO'$ for Player~$ O $ maps a play prefix~$
r_{0} q_{0} \cdots r_{i} $ to a state~$ q_{i} \in
\dom(r_{i}) $. A play $ r_{0} q_{0} r_{1} q_{1} r_{2} q_{2} \cdots $
is consistent with $ \stratI' $, if $ r_{i} = \stratI'(r_{0} q_{0}
\cdots r_{i-1} q_{i-1}) $ for every $ i \in \nats $ and it is
consistent with $ \stratO' $, if $ q_{i} = \stratO'(r_{0} q_{0} \cdots
r_i) $ for every $ i \in \nats $. A strategy $ \strat' $ for Player~$\p\in \set{I,O}$ is winning, if every play that is consistent with $\strat'$ is winning for Player~$\p$. As usual, we say that a player wins $ \game(\aut) $, if she has a winning strategy.

 First, we show that the original delay game and the abstract game~$\game(\aut)$ have the same winner. To this end, we fix a winning strategy for Player~$O$ in the delay game and simulate it in $\game(\aut)$, and vice versa, by translating real moves of Player~$I$ into abstract moves from $\curlyR$, and vice versa.

\begin{lem} \label{lem_splitgamecorrectness} Player~$ O $ wins $
  \delaygame{\lpar{\aut}} $ for some delay function $ f $ if and only
  if Player~$ O $ wins $ \game(\aut) $.
\end{lem}

\begin{proof}
For the sake of readability, we denote $\game(\aut)$ by $\game$ and $\delaygame{\lpar{\aut}}$ by $\Gamma$, whenever $f$ is clear from the context.
 
\myquot{$ \Rightarrow $}: Let $ \stratO $ be a winning strategy for
  Player~$ O $ in $ \delaygame{\lpar{\aut}} $ for some $f$, which we can
  assume\footnote{This assumption simplifies the proof, but can be
    avoided at the expense of introducing intricate notation.} to be a constant~\cite{HoltmannKaiserThomas12}. We construct a
  winning strategy~$ \stratO' $ for Player~$ O $ in $ \game $ by simulating a play of $\game$ by a play of $\Gamma$.

  Let $ r_{0} $ be the first move of Player~$ I $ in $ \game $,
  which has to be answered by Player~$ O $ by picking $
  \stratO'(r_{0}) = q_{I}^{ \autc } $, and let $ r_{1} $ be Player~$ I $'s
  response. As $\wit{\resolve_0}$ and $\wit{\resolve_1}$ are infinite by definition, we
  can choose witnesses $ w_{0} \in \wit{\resolve_0} $ and $w_1 \in \wit{\resolve_1}$ such that $f(0) \le \size{w_{0}} \le \size{w_{1}} $. We simulate this play prefix in $\Gamma$: Player~$I$ picks $ w_{0} w_{1} =
  \widx{\alpha}{0} \cdots \widx{\alpha}{\ell_{1} - 1} $ in his first
  moves, which is enough to play at least $\size{w_0}$ rounds: $w_0$ is long enough to play the first round and $w_1$ is even long enough for at least the next $\size{w_0}$ rounds. 
  
Now, let $ \widx{\beta}{0} \cdots \widx{\beta}{\ell_{1} - f(0)}
  $ be the response of Player~$ O $ according to $ \stratO $. As we have played at least $\size{w_0}$ rounds, we obtain $ \size{\widx{\beta}{0} \cdots
    \widx{\beta}{\ell_{1} - f(0)}} \geq \size{w_{0}} $. Thus, we are
  in the following situation for $ i = 1$:

  \begin{itemize}
    
  \item in $ \game $, we have constructed a play prefix~$ r_{0}
    q_{0} \cdots r_{i-1} q_{i-1} r_{i} $,  

  \item in $ \Gamma $, Player~$ I $ has picked $ w_{0}
    \cdots w_{i} = \widx{\alpha}{0} \cdots
    \widx{\alpha}{\ell_{i} - 1} $ and Player~$O$ has picked $
    \widx{\beta}{0} \cdots \widx{\beta}{\ell_{i} - f(0)} $ according
    to $\stratO$. The moves of the players satisfy $ \size{\widx{\beta}{0} \cdots
      \widx{\beta}{\ell_{i} - f(0)}} \geq \size{w_{0} \cdots w_{i-1}}
    $.
	
	\item Finally, $ w_{j} $ is a witness for $ r_{j}
    $ for every $ j \leq i $. 

  \end{itemize}
  
  Now, let $i \geq 1$ be arbitrary and let $q_{i-1} = (q_{i-1}', c_{i-1})$. We define $ q_{i} $ to be the
  state of $ \autc $ that is reached from $ (q_{i-1}', \col(q_{i-1}')) $ after
  processing $ w_{i-1} $ and the corresponding moves of Player~$ O $,
  i.e.,
  \begin{equation*}
    { \widx{\alpha}{\size{w_{0} \cdots w_{i-2}}} \choose 
        \widx{\beta}{\size{w_{0} \cdots w_{i-2}}} } 
      \cdots {\widx{\alpha}{\size{w_{0} \cdots w_{i-1}}-1} 
        \choose \widx{\beta}{\size{w_{0} \cdots
            w_{i-1}}-1}}.
  \end{equation*}
  By definition of $ r_{i-1}$, we have $ q_{i} \in
  \resolve_{i-1}(q_{i-1})$. Correspondingly, we can define \[
  \stratO'(r_{0} q_{0} \cdots r_{i-1} q_{i-1} r_{i}) = q_{i} .\] Now,
  let $ r_{i+1} $ be the next move of Player~$ I $ in $ \game $
  and let $ w_{i+1} \in \wit{r_{i+1}} $ be a witness with $ |w_{i+1} |
  \geq \size{w_i} $. In $ \Gamma $, let Player~$ I $ pick $
  w_{i+1} = \widx{\alpha}{\ell_{i}} \cdots \widx{\alpha}{\ell_{i+1} -
    1} $ as his next moves and let Player~$ O $ respond by $
  \widx{\beta}{\ell_{i} - f(0) + 1} \cdots \widx{\beta}{\ell_{i+1} -
    f(0) } $ according to $ \stratO $. Thus, we are again in the
  aforementioned situation for $ i + 1 $, which concludes the
  definition of $ \stratO' $. 

  It remains to show that $ \stratO' $ is winning. Consider a play~$ r_{0} q_{0} r_{1}
  q_{1} r_{2} q_{2} \cdots $ that is consistent with $ \stratO' $ and
  let $ w = { \widx{\alpha}{0} \choose \widx{\beta}{0} }{
    \widx{\alpha}{1} \choose \widx{\beta}{1} }{ \widx{\alpha}{2}
    \choose \widx{\beta}{2} } \cdots $ be the outcome in $
  \Gamma $ constructed during the simulation as defined above. Note that $\widx{\alpha}{0} \widx{\alpha}{1} \widx{\alpha}{2} \cdots$ is equal to $w_0 w_1 w_2 \cdots$, where each $w_i$ is a witness of $\resolve_i$. Furthermore,  let $q_i = (q_i', c_i)$. 
  
  A straightforward inductive application of Remark~\ref{remark_autproduct} shows that $q_{i+1}'$ is the state that $\aut$ reaches after processing $w_{i}$ and the corresponding moves of Player~$O$ starting in $q_{i}'$ and that $c_{i+1}$ is the maximal color seen on the run.
Thus, the
  maximal color visited infinitely often by $ \aut $ after processing $ w $ is
  the same as the maximal color of the sequence $
  c_{0}c_{1}c_{2}\cdots $, which is even, as $ w $ is consistent with a winning strategy and therefore accepted by $\aut$. Hence, $r_{0} q_{0} r_{1}
  q_{1} r_{2} q_{2} \cdots$ is winning for Player~$O$ and $\stratO'$ is a winning strategy.

  \myquot{$ \Leftarrow $}: Let $
  \stratO' $ be a wining strategy for Player~$ O $ in $ \game
  $. We construct a winning strategy~$\stratO$ for her in $\delaygame{\lpar{\aut}}$ for the constant delay function~$f$ with $f(0) = 2d$, where $ d = 2^{n^{2}} $. The strategy~$\stratO$ is again constructed by simulating a play of $\Gamma$ by a play of $\game$.
  
   In the following, both players pick their moves in $\Gamma$ in blocks of length~$d$. We denote Player~$I$'s blocks by $\block{a_i}$ and Player~$O$'s blocks by $\block{b_i}$, i.e.,  every $\block{a_i}$ is in $\SigmaI^d$ and every $\block{b_i}$ is in $\SigmaO^d$. 
   
  Let $ \block{a_{0}}\block{a_{1}} $ be the first move of Player~$ I $ in $
  \Gamma $, define $q_0 = q_I^{ \autc }$, and let $
  r_{0} = \resolve_{\block{a_{0}}}^{\set{q_{0}}} $ and $ r_{1} =
  \resolve^{r_{0}(q_{0})}_{\block{a_{1}}} $. Then, $ r_{0} q_{0} r_{1} $ is a play prefix in $ \game $ that
  is consistent with $ \stratO' $. Thus, we are in the following
  situation for $ i = 1 $:

  \begin{itemize}
    
    \item in $ \Gamma $, Player~$ I $ has picked blocks~$
      \block{a_{0}} \cdots \block{a_{i}} $ and Player~$ O $ has picked $
      \block{b_{0}} \cdots \block{b_{i-2}} $,
	
  \item in $ \game $ we have constructed a play prefix~$ r_{0}
    q_{0} \cdots r_{i-1} q_{i-1} r_{i} $ that is consistent with $
    \stratO' $, and
	
	\item $ \block{a_{j}} $ is a witness for $
      r_{j} $ for every $ j \leq i $.

  \end{itemize}

  Now, let $ i \geq 1 $ be arbitrary and $ q_i =
  \stratO'(r_{0} q_{0} \cdots r_{i-1} q_{i-1} r_i) $. The rules of $\game$ imply $ q_{i} \in
  \dom(r_{i}) = r_{i-1}(q_{i-1}) $. Furthermore, as $ \block{a_{i-1}} $ is a witness for
  $ r_{i-1} $, there is some $ \block{b_{i-1}} $ such that the automaton $ \autc $ reaches $ q_{i} $ after
  processing $ { \block{a_{i-1}} \choose \block{b_{i-1}} }$ from $
  (q_{i-1}',\col(q_{i-1}')) $, where $ q_{i-1} =
  (q_{i-1}', c_{i-1}) $. Player~$ O $'s strategy for $
  \Gamma $ is to pick the letters of $ \block{b_{i-1}} $
  in the next $ d $ rounds. These are answered by Player~$ I $ by $ d $
  letters forming $ \block{a_{i+1}} $. This way, we obtain $ \resolve_{i+1} =
  \resolve^{\resolve_{i}(q_{i})}_{\block{a_{i+1}}} $ bringing us back to the
  aforementioned situation for $ i + 1 $, which concludes the definition of $\stratO$.

  It remains to show that $ \stratO $ is winning for Player~$ O $. Let
  $ w = {\block{a_{0}} \choose \block{b_{0}}}{\block{a_{1}} \choose \block{b_{1}}}{\block{a_{2}} \choose
    \block{b_{2}}} \cdots $ be the outcome of a play of $ \Gamma $
  consistent with $ \stratO $. Also, let $ r_{0} q_{0} r_{1} q_{1} r_{2} q_{2}
  \cdots $ be the corresponding play of $ \game $ constructed as described in the simulation above, where each $\block{a_i}$ is a witness for $\resolve_i$. Finally, let  $q_i = (q_i', c_i)$. 
  
  A straightforward inductive application of Remark~\ref{remark_autproduct} shows that $ q_{i+1}' $ is the state
  reached in $ \aut $ after processing $ {\block{a_{i}} \choose \block{b_{i}}} $ 
  starting in $ q_{i}' $ and that $ c_{i+1} $ is the largest color seen
  on this run.  As $ r_{0} q_{0} r_{1} q_{1}
  r_{2} q_{2} \cdots $ is consistent with $ \stratO' $, the sequence~$ \col_\autc(q_0) \col_\autc(q_1) \col_\autc(q_2) \cdots = 
  c_{0}c_{1}c_{2}\cdots $ satisfies the parity condition, i.e., the
  maximal color occurring infinitely often is even. Thus, the maximal
  color occurring infinitely often during the run of $\aut $ on $w$ is even as well, i.e., $w$ is winning for Player~$O$. Thus, $ \stratO $ is a winning strategy for Player~$ O $.
\end{proof}

Now, we can prove the main theorem of this section: determining whether Player~$O$ wins a delay game induced by a given parity automaton $\aut$ for some $f$ is in $\exptime$. To this end, it suffices to model $\game(\aut)$ as a classical parity game and to show that it can be constructed and solved in exponential time. This is sufficient due to Lemma~\ref{lem_splitgamecorrectness}.

\begin{proof}[Proof of Theorem~\ref{thm_paritycomplexity_ub}]  
  First, we argue that $\curlyR$ can be constructed in exponential time: to this end, one constructs for every partial function~$\resolve$ from $Q_{ \autc }$ to $2^{Q_{ \autc }}$ the automaton of Lemma~\ref{lemma_graphautprops}.\ref{lemma_graphautprops_aut} recognizing~$\wit{\resolve}$ and tests it for recognizing an infinite language. There are exponentially many functions and each automaton is of exponential size, which yields the desired result.
  
  Now, we can encode $ \game(\aut) $ as a
  graph-based parity game\footnote{For a complete definition see,
  e.g., \cite{GraedelThomasWilke02}}~$ ((V, V_{I}, V_{O}, E), \col') $ with $
  \col' \colon V \rightarrow\nats $ where

  \begin{itemize}
  
  \item $ V = V_{I} \cup V_{O} $,
  
  \item $V_{I} =
    \set{\init} \cup \curlyR \times Q_{ \autc } $ for some fresh initial vertex~$\init$,
    \item $V_{O} = \curlyR $,
	
	\item $E$ is the union of the following sets of edges:
	\begin{itemize}
		
		\item $\set{(\init, \resolve) \mid \dom(\resolve) = \set{q_I^{ \autc }}}$: the initial moves of Player~$I$.
		
		\item $\set{((\resolve, q), \resolve') \mid \dom(\resolve') = \resolve(q)}$: (regular) moves of Player~$I$.
		
		\item $\set{(\resolve,(\resolve, {q})) \mid q \in \dom(\resolve)}$: moves of Player~$O$, and
\end{itemize}	
	\item $\col'(v)  = \begin{cases}
	c &\text{if } v = (r,(q,c)) \in \curlyR \times Q_{ \autc },\\
	0 &\text{otherwise.} 
	
\end{cases}$       
  \end{itemize}

  Then, Player~$ O $ wins $ \game(\aut) $ if and only if she
  has a winning strategy from $ \init $ in the explicit parity game. 
  
  A parity game with $ n $ vertices, $ m $ edges, and $ k $ colors can
  be solved in time~$ \bigo(m  n^{\frac{k}{3}}) $~\cite{Schewe07}. The parity game constructed above has at most $ \bigo(2^{\size{\aut}^4} \cdot \size{\aut}^2 ) $ vertices and at most $\size{\aut}$ colors. Hence,
  it can be solved in exponential time in the size of
  $\aut$.
\end{proof}

By applying
both directions of the equivalence between $ \delaygame{\lpar{\aut}} $
and $ \game(\aut) $, we obtain an exponential upper bound of the form~$2^{(\size{\aut} k)^2+1}$ on the lookahead necessary
for Player~$ O $ to win a delay game with a parity condition over $k$ colors. However, this upper bound on the lookahead requires the introduction of the game~$\game(\aut)$. To conclude this section, we present an alternative proof yielding a slightly better upper bound on the necessary lookahead for Player~$O$ in case she wins. The proof relies on a direct pumping argument for winning strategies for Player~$I$: winning for the constant delay function~$f$ with $ f(0) \leq  2^{\size{\aut}^2 k +1}$ allows him to cover arbitrarily larger lookahead\footnote{The previous version of this paper claimed a bound of $ 2^{\size{\aut} k +1}$, but the proof was incomplete. Here, we present a corrected proof for the bound~$2^{\size{\aut}^2 k +1}$.}. 

\begin{thm} \label{thm_upperbound_direct} Let $ L = \lpar{\aut} $
  where $ \aut $ is a deterministic parity automaton with $ k $
  colors. The following are equivalent:

  \begin{enumerate}

  \item \label{thm_upperbound_winsspecial} Player~$ O $ wins $
    \delaygame{L} $ for some constant delay function $ f $ with $ f(0) \leq  2^{\size{\aut}^2 k +1}$.

  \item \label{thm_upperbound_winsany} Player~$ O $ wins $
    \delaygamep{L} $ for some delay function $ f' $.

  \end{enumerate}
\end{thm}

\begin{proof}
We only consider the non-trivial implication~(\ref{thm_upperbound_winsany})~$\Rightarrow$~(\ref{thm_upperbound_winsspecial}) and prove the contrapositive by turning a winning strategy for Player~$ I $ in
  $ \delaygame{L} $ into a winning strategy for games with
  arbitrarily larger lookahead. We use a pumping argument to show that every sufficiently large lookahead contains
  some repetitive behavior, which we pump to cover the larger lookahead. To this end, we construct a
  simulation of the original game to inductively define
  the new strategy for Player~$ I $. The simulation proceeds in big
  steps, i.e., the players pick sufficiently large blocks of input and output letters,
  respectively, which allows us to identify the repetitive behaviour
  in each such block. The structure of the simulation is similar to the game~$\game(\aut)$ from the proof
  of Theorem~\ref{thm_paritycomplexity_ub}, where Player~$ I $ is always two functions ahead, after he moved, while Player~$ O $
  responds only to the first of these two functions. In this alternative proof, we do
  not need an abstract representation of  the
  lookahead, which is why we represent it just by concrete sequences of input letters. 
  
   As in the previous proof, simulating the game in big steps requires us to extend the automaton recognizing
  $ L $ to accumulate the maximal color along a run.  Let $ \aut = (Q,\SigmaI \times \SigmaO, q_{I},
  \delta, \col) $ be given and let
  $ \autc = (Q_{ \autc },\SigmaI \times \SigmaO, q_{I}^{ \autc },
  \delta_{ \autc }, \col_{ \autc }) $
  be the color-tracking extension of $ \aut $ as defined before. 
    
Now, let $\stratI$ be a winning strategy for Player~$I$ for $\delaygame{L}$, let $f'$ be an arbitrary delay function. Due to the results of Holtmann et al.~\cite{HoltmannKaiserThomas12}, we can assume w.l.o.g.\ that $f'$ is constant, which simplifies the proof, but is not essential. We define $d = 2^{\size{\aut}^2\cdot k}$, i.e., we have $f(0) = 2d$. We construct a winning strategy~$\stratI'$ for Player~$I$ in $\delaygamep{L}$ by simulating a play of $\delaygamep{L}$ via a play of $\delaygame{L}$.
In the following, we will group the moves of the players in $\delaygame{L}$ into blocks of length~$d$. We denote Player~$I$'s blocks by $\block{a_i}$ and Player~$O$'s blocks by $\block{b_i}$. For the sake of readability, we denote $\delaygame{L}$ by $\Gamma$ and $\delaygamep{L}$ by $\Gamma'$.

To begin, we have to give the first moves of Player~$I$ in $\Gamma'$.
To this end, consider the first move $ \stratI(\epsilon) = \block{a_0} \block{a_1}$ of Player~$I$ in  $\Gamma$. By the choice of $d$ and the fact that $\aut$ is complete, we can decompose $ \block{a_{0}} $ into $x_{0}y_{0}z_{0} $ with non-empty $y_0$ such that for every $q \in Q$
\[
\delta_{\autp}^*(\set{(q, \col(q))}, x_0) = \delta_{\autp}^*(\set{(q, \col(q))}, x_0 y_0).
\] Similarly, we can decompose $ \block{a_{1}} $ into $x_{1}y_{1}z_{1} $ with non-empty $y_1$ such that for every $q \in Q$ 
\[
\delta_{\autp}^*(\set{(q, \col(q))}, x_1) = \delta_{\autp}^*(\set{(q, \col(q))}, x_1 y_1).
\]
Now, we pump $ y_{0}
  $ sufficiently often to simulate the initial block in $\Gamma'$, i.e., let $\widx{\alpha}{0} \cdots \widx{\alpha}{\ell_0-1} = x_{0}(y_{0})^{h_{0}}z_{0} $ for the smallest $ h_{0} \ge 1$ such that $ \size{x_{0}(y_{0})^{h_{0}}z_{0}} \geq f'(0) $. Similarly, we define  $\widx{\alpha}{\ell_0} \cdots \widx{\alpha}{\ell_1-1} = x_{1}(y_{1})^{h_{1}}z_{1} $ for the smallest $ h_{1} \ge 1 $ with $ \size{x_{1}(y_{1})^{h_{1}}z_{1}} \geq \size{x_{0}(y_{0})^{h_{0}}z_{0}} $.

Now, we define $\stratI'$ such that Player~$I$ picks $\widx{\alpha}{0} \cdots \widx{\alpha}{\ell_1-1} $ with his first moves in $\Gamma'$, which is answered by Player~$O$ by picking $\widx{\beta}{0} \cdots \widx{\beta}{\ell_1 -f'(0)}$. By the choice of $h_1$, we obtain $(\ell_1 -f'(0)) +1 \geq \ell_0$. 
  
Now, we are in the following situation for $ i = 1 $:

  \begin{itemize}
    
    \item In $ \Gamma', $ Player~$ I $ has picked $
      \widx{\alpha}{0} \cdots \widx{\alpha}{\ell_i-1}
      $ such that for every $j \le i$: $\widx{\alpha}{\ell_{j-1}} \cdots \widx{\alpha}{\ell_{j}-1} = x_{j}(y_{j})^{h_{j}}z_{j} $.

	  \item In $ \Gamma'$, Player~$ O $ has picked $
      \widx{\beta}{0} \cdots \widx{\beta}{\ell_i-f'(0)}
      $ such that $(\ell_i -f'(0)) +1 \geq \ell_{i-1}$. Thus, Player~$O$ has  provided an answer to $\widx{\alpha}{\ell_{i-2}} \cdots \widx{\alpha}{\ell_{i-1}-1} = x_{i-1} (y_{i-1})^{h_{i-1}}z_{i-1}$.

  \item In $ \Gamma $, Player~$ I $ has picked blocks~$ \block{a_0} \cdots 	
  	\block{a_i}$ and Player~$ O $ has picked $ \block{b_0} \cdots \block{b_{i-2}} $.

  \end{itemize}
  
Now, let $i \geq 1$ be arbitrary and let $q_{i-1}$ be the state reached by $\aut$ after processing ${ \block{a_0} \choose \block{b_0} } \cdots { \block{a_{i-2}} \choose \block{b_{i-2}} }$. Further, 
 let $(q_{i-1}^*,c_{i-1}^*)$ be the state reached by $\autc$ after processing $x_{i-1}(y_{i-1})^{h_{i-1}-1}$ and the corresponding letters picked by Player~$O$ at these positions starting in $(q_{i-1},\col(q_{i-1}))$. Due to the state repetition induced by the decomposition, which is independent of the starting state~$(q,\col(q))$, we can find a word~$x_{i-1}' \in \SigmaO^{\size{x_{i-1}}}$ such that $\autc$ reaches the same state~$(q_{i-1}^*,c_{i-1}^*)$ after processing ${x_{i-1} \choose x_{i-1}'}$, when starting in $(q_{i-1}, \col(q_{i-1}))$. 
We define $\block{b_{i-1}} = x_{i-1}' y_{i-1}' z_{i-1}'$, where $y_{i-1}'$ and $z_{i-1}'$ are the letters picked by Player~$O$ at the positions of the last repetition of $y_{i-1}$ and at the positions of $z_{i-1}$, respectively. 

We continue the simulation in $\Gamma$ by letting Player~$O$ pick the block $\block{b_{i-1}}$ during the next $d$ rounds, which is answered by Player~$I$ by picking the next block $\block{a_{i+1}}$. Again, we can decompose $ \block{a_{i+1}} $ into $x_{i+1}y_{i+1}z_{i+1} $ with non-empty $y_{i+1}$ such that for every $q \in Q$
 \[
\delta_{\autp}^*(\set{(q, \col(q))}, x_1) = \delta_{\autp}^*(\set{(q, \col(q))}, x_1 y_1).
\]
As before, we define  
   $\widx{\alpha}{\ell_i} \cdots \widx{\alpha}{\ell_{i+1}-1} = x_{i+1}(y_{i+1})^{h_{i+1}}z_{i+1} $ for the smallest $ h_{i+1} \ge 1 $ such that $ \size{x_{i+1}(y_{i+1})^{h_{i+1}}z_{i+1}}
   \geq  \size{x_{i}(y_{i})^{h_{i}}z_{i}} $.
  Now, we define $\stratI'$ such that Player~$I$ picks $\widx{\alpha}{\ell_i} \cdots \widx{\alpha}{\ell_{i+1}-1}$ with his next moves in $\Gamma'$, which is answered by Player~$O$ by picking $\widx{\beta}{(\ell_i-f'(0))+1} \cdots \widx{\beta}{\ell_{i+1} -f'(0)}$. By the choice of $h_{i+1}$, we obtain $(\ell_{i+1} -f'(0)) +1 \geq \ell_i$. 
  Thus, we are in the situation described above for $i+1$, which completes the definition of $\stratI'$.
    
  It remains to show that $\stratI'$ is winning. Let $w' = 
  { \widx{\alpha}{0} \choose \widx{\beta}{0} }
  { \widx{\alpha}{1} \choose \widx{\beta}{1} }
  { \widx{\alpha}{2} \choose \widx{\beta}{2} } \cdots $ be an outcome of a play that is consistent with $\stratI'$ in $\Gamma'$ and let 
  $ w =  {\block{a_0} \choose \block{b_0} }
  {\block{a_1} \choose \block{b_1} }
  {\block{a_2} \choose \block{b_2} } \cdots $ be the simulating play of $\Gamma$ as described above. Fix $(q_0,c_0) = (q_I,\col(q_I))$ and let $(q_{i}, c_{i})$ for $i \ge 1$ be the state reached by $\autc$ when processing  
\[
{ \widx{\alpha}{\ell_{i-2}} \choose \widx{\beta}{\ell_{i-2}} } \cdots
{\widx{\alpha}{\ell_{i-1}-1} \choose \widx{\beta}{\ell_{i-1}-1} },
\]  
starting in $(q_{i-1}, \col(q_{i-1}))$, using $\ell_{-1} = 0$. By construction, $(q_{i}, c_{i})$ is also the state reached by $\autc$ when processing 
$  {\block{a_{i-1}} \choose \block{b_{i-1}} }$
starting in $(q_{i-1}, \col(q_{i-1}))$. Thus, applying Remark~\ref{remark_autproduct} inductively shows that $\aut$ accepts $w$ if, and only if, it accepts $w'$. Since $w$ is consistent with the winning strategy~$\stratI$ both are not accepted, i.e., $\stratI'$ is indeed a winning strategy.
\end{proof}

To conclude, let us briefly discuss the gap in the argumentation of the original proof. There, when decomposing the second block~$\block{a_1}$, we only asked for a state repetition from the state of the powerset automaton reached after processing the first block~$\block{a_0}$. However, this does not guarantee that the state~$(q_{1}^*, c_{1}^*)$ as defined above is reachable from $(q_1,\col(q_1))$, the state reached after processing ${\block{a_0} \choose \block{b_0}}$ from $(q_I,\col(q_I))$. One only obtains that there is a block $\block{b_0'}$ such that $(q_{1}^*, c_{1}^*)$ is reachable from the state reached after processing ${\block{a_0} \choose \block{b_0'}}$ from $(q_I,\col(q_I))$. Hence, if $\block{b_0'} \neq \block{b_0}$, then the inductive argument breaks down. Having a state repetition for every possible such starting state resolves this issue.


\section{Winning Conditions that are Reachability and Safety}
\label{sec_safetycapreach}
Recall the exponential lower bounds on the necessary lookahead for reachability and safety conditions presented in Section~\ref{sec_lowerbounds}. Both rely on the same construction and one can even turn the deterministic safety automata~$\aut_n'$ exhibiting the lower bound into reachability automata exhibiting the same lower bound. To this end, one changes the set of accepting states, but leaves the transition structure of $\aut_n'$ unchanged. 
Nevertheless, these automata do not accept the same language. In the following, we prove that this is unavoidable: linear lookahead is necessary and sufficient for winning conditions that are both reachability and safety. Furthermore, solving such games is $\conp$-complete and thus simpler than solving delay games with general reachability or safety conditions (under standard complexity-theoretic assumptions).

Winning conditions that are both reachability and safety conditions are also known as $\omega$-regular clopen conditions, as they are those $\omega$-regular languages that are both closed and open in the Cantor topology. Note that the way we represent such conditions is a non-trivial issue that has influences on our results: we could specify an $\omega$-regular clopen language either by a reachability or by a safety automaton or even give a pair of equivalent automata, one for each acceptance condition. However, we take another approach. A clopen language is fully characterized by a finite language of finite words. Thus, we use acyclic finite automata on finite words to represent $\omega$-regular clopen winning conditions. For every deterministic reachability and every deterministic safety automaton recognizing a clopen language there is an acyclic finite automaton of the same size that represents the same language, i.e., the change in representation does not incur a blowup, when starting with deterministic automata. The results mentioned above pertain to the representation of winning conditions via acyclic finite automata. 

First, we formalize the characterization of $\omega$-regular clopen languages by acyclic finite automata. In the following, we only consider (w.l.o.g.)~automata whose states are all reachable from the initial state. Also, recall that we require all our automata to be complete. A state of an automaton is productive, if an accepting state is reachable from it, otherwise it is non-productive. In particular, every non-productive state is non-accepting. Finite languages of finite words are recognized by acyclic finite automata, i.e., automata that satisfy the following property: if $q$ is on a cycle, then $q$ is non-productive. The depth of an acyclic automaton~$\aut$ is the length of the longest path from an initial to an accepting state, which is bounded by $\size{\aut}-1$.

\begin{rem}\label{rem_acyclic}
Let $\aut$ be a non-deterministic automaton. Then, $\lstar{\aut}$ is finite if and only if $\aut$ is acyclic.
\end{rem}

By merging all non-productive states of a non-deterministic acyclic automaton into a single non-accepting sink, we obtain an equivalent acyclic automaton of the same depth where processing a word that is longer than the depth leads to the non-accepting sink.

\begin{rem}
Let $\aut$ be a non-deterministic acyclic automaton of depth~$k$. Then, there is a non-deterministic acyclic automaton~$\aut'$ of depth~$k$ with $\size{\aut'} \le \size{\aut}$, $\lstar{\aut} = \lstar{\aut'}$, and such that processing a word of length greater than $k$ leads $\aut'$ to a (non-accepting) sink.	
\end{rem}

Now, we are able to state the characterization theorem for $\omega$-regular clopen languages. 

\begin{lem}[\cite{KupfermanVardi01}]
\label{lemma_reachabilitycapsafety_charac}
Let $L \subseteq \Sigma^\omega$ be an $\omega$-regular clopen language that is recognized by a deterministic automaton~$\aut$ with reachability or safety acceptance.  
 Then, $ L = \lstar{\aut'} \cdot \Sigma^{\omega} $ for some deterministic acyclic automaton~$\aut'$ with $\size{\aut'} \le \size{\aut}$.
\end{lem}

With this characterization at hand, we determine the complexity of solving  delay games with $\omega$-regular clopen winning conditions and then present tight upper and lower bounds on the necessary lookahead.

\begin{thm}\label{thm_reachabilitycapsafetycomplexity}
The following problem is $\conp$-complete: Given an acyclic automaton~$\aut$ over $\SigmaI \times \SigmaO$, does Player~$O$ win $\delaygame{\lstar{\aut} \cdot (\SigmaI \times \SigmaO)^\omega}$ for some $f$?
\end{thm}

The proof of the preceding theorem is split into the following two lemmata. We begin by showing $\conp$-hardness, which already holds for deterministic automata. Then, we show $\conp$-membership, which also holds for non-deterministic automata. 

 The  following hardness proof is reminiscent of the proof of Theorem~\ref{thm_safetyhardness}, but simpler since we only have to deal with universal polynomial time machines instead of alternating polynomial space machines.

\begin{lem}
The following problem is $\conp$-hard: Given an acyclic deterministic automaton~$\aut$ over $\SigmaI \times \SigmaO$, does Player~$O$ win $\delaygame{\lstar{\aut} \cdot (\SigmaI \times \SigmaO)^\omega}$ for some $f$?
\end{lem}

\begin{proof}
We show hardness by a reduction from the acceptance problem for universal polynomial time Turing machines. Let $ \tm = (Q, \Sigma, q_I, \Delta, \qacc, \qrej)$  be such a machine, where $\Delta \subseteq Q \times \Sigma \times Q \times \Sigma \times \set{-1, 0, 1}$ is the transition relation. We assume w.l.o.g.\ that the accepting state~$\qacc$ and the rejecting state~$\qrej$ have self-loops. Let $p$ be a polynomial that bounds the time-consumption of $\tm$. Furthermore, let $x \in \Sigma^*$ be an input for $\tm$.

We construct an acyclic deterministic automaton~$\aut$ over $\SigmaI \times \SigmaO$ such that Player~$O$ wins $\delaygame{\lstar{\aut} \cdot (\SigmaI \times \SigmaO)^\omega}$ for some $f$ if and only if $\tm$ accepts $x$, i.e., if and only if every run of $\tm$ on $x$ is accepting. To this end, we construct a game where Player~$I$ produces sequences of configurations to simulate runs of $\tm$ on $x$ and Player~$O$ can claim errors in this simulation. Player~$O$ wins, if she correctly claims an error or if the simulated run is accepting. Thus, we define $\SigmaI = Q \cup \Sigma$ and $\SigmaO= \set{ \noerror, \copyerror }$. Player~$O$ uses $\copyerror$ to claim an error.

Let $t = p(\size{x}) $ and let $\conf$ be the set of encodings of configurations of $\tm$ of length~$t$, i.e., words over the alphabet $\Sigma \cup Q$ of length~$t +1 $ containing exactly one letter from $Q$. Furthermore, let $c_0 \in \conf$ denote the encoding of the initial configuration of $\tm$ on $x$. In the following, we do not distinguish between a configuration and its encoding. 

Now, consider 
the language~$L$ of words~${u \choose v} \in (\SigmaI \times \SigmaO)^*$ such that either 
\[u \in \SigmaI^{\,t \cdot (t+1)} \setminus c_0 \cdot (\conf) ^{t-1} \]
or, if $u$ is in $c_0 \cdot (\conf)^{t-1}$, then $u$ either ends with an accepting configuration or the smallest position marked with $\copyerror$ in $v$ indicates a cell which witnesses that the following configuration in $u$ is not a successor configuration of the current one. It is straightforward to show that $L$ is recognized by an automaton~$\aut$ of polynomial size in $\size{\tm} + t$, which is polynomial in the size of $\tm$ and the length of $x$. To this end, one relies on the fact that $L$ only contains words of a fixed length, which is polynomial in $t$, and on the error checking routine described in the proof of Theorem~\ref{thm_safetyhardness}. Finally, as $L$ is finite, Remark~\ref{rem_acyclic} implies that $\aut$ is acyclic.

It remains to show that Player~$O$ wins $\delaygame{L \cdot (\SigmaI \times \SigmaO)^\omega}$ for some delay function~$f$ if and only if $\tm$ accepts $x$, i.e., every run of $\tm$ on $x$ is accepting.

Thus, let $x$ be accepted by $\tm$, i.e., every run of $\tm$ on $x$ reaches an accepting configuration after executing at most $t-1$ transitions. Hence, in order to win a delay game with winning condition~$L \cdot (\SigmaI \times \SigmaO)^\omega$, Player~$I$ has to introduce an error in the simulation of $\tm$ on $x$. With sufficiently large constant lookahead, Player~$O$ can catch him and correctly claim the error. Hence, she wins $\delaygame{L \cdot (\SigmaI \times \SigmaO)^\omega}$ for some $f$.

Conversely, assume $x$ is not accepted by $\tm$, i.e., there is a run of $\tm$ that does not reach an accepting configuration. Player~$I$ can simulate this run in $\delaygame{L \cdot (\SigmaI \times \SigmaO)^\omega}$ for every~$f$. As he does not introduce an error, he wins the resulting play, i.e., Player~$O$ does not win $\delaygame{L \cdot (\SigmaI \times \SigmaO)^\omega}$ for any $f$.
\end{proof}

Next, we show that the problem is in $\conp$, which completes the proof of Theorem~\ref{thm_reachabilitycapsafetycomplexity}

\begin{lem}
\label{lem_reachabilitycapsafetyconp}
The following problem is in $\conp$: Given an acyclic non-deterministic automaton~$\aut$ over $\SigmaI \times \SigmaO$, does Player~$O$ win $\delaygame{\lstar{\aut} \cdot (\SigmaI \times \SigmaO)^\omega}$ for some $f$?
\end{lem}

\begin{proof}
Let $L = \lstar{\aut} \cdot (\SigmaI \times \SigmaO)^\omega$, which is in particular a reachability condition. Thus, Theorem~\ref{theorem_reachcomplex} shows that Player~$O$ wins $\delaygame{L}$ if and only if $\proj{L}$ is universal. Let $k$ be the depth of $\aut$. As every word of length greater than $k$ leads to a non-accepting sink state of $\aut$, membership in $\lstar{\aut} \cdot (\SigmaI \times \SigmaO)^\omega$ only depends on the prefix of length~$k$. Hence, $\proj{L}$ is universal if and only if  $\proj{\aut}$ accepts at least one prefix of every word of length~$k \le \size{\aut}$. The latter problem can be solved in polynomial time by a universal Turing machine.
\end{proof}

After resolving the complexity of delay games with $\omega$-regular clopen winning conditions, we turn our attention to proving lower and upper bounds on the necessary lookahead in such games. We begin by showing that the depth of the automaton recognizing the winning condition is an upper bound on the necessary lookahead, which follows from the reasoning presented in the proof of Lemma~\ref{lem_reachabilitycapsafetyconp}.

\begin{thm} \label{thm_reachabilitycapsafetybounds}

  Let $\aut$ be an acyclic non-deterministic automaton over $\SigmaI \times \SigmaO$ and let $k$ be the depth of $\aut$. The following are
  equivalent:

  \begin{enumerate}

  \item \label{thm_reachabilitycapsafetybounds_any}
   Player~$ O $ wins $ \delaygame{\lstar{\aut} \cdot (\SigmaI \times \SigmaO)^\omega} $ for some delay function~$
    f $.

  \item\label{thm_reachabilitycapsafetybounds_linear}
   Player~$ O $ wins $ \delaygame{\lstar{\aut} \cdot (\SigmaI \times \SigmaO)^\omega}$ for some constant delay
    function~$ f $ with $ f(0) \leq k $.
  \end{enumerate}
\end{thm}

\begin{proof}
We only consider the non-trivial implication~(\ref{thm_reachabilitycapsafetybounds_any})~$\Rightarrow$~(\ref{thm_reachabilitycapsafetybounds_linear}). Let $L = \lstar{\aut} \cdot (\SigmaI \times \SigmaO)^\omega$.  

Assume Player~$O$ wins $\delaygame{L}$ for some $f$. Then, by Theorem~\ref{theorem_reachcomplex}, $\proj{L}$ is universal, which implies that $\proj{L(\aut)}$ contains at least one prefix of every word of length $k$. Hence, for every possible first move~$\alpha(0) \cdots \alpha(k-1) \in \SigmaI^k$ of Player~$I$ in $\delaygame{L}$, there is a $\beta(0) \cdots \beta(k-1) \in \SigmaO^k$ such that ${\alpha(x) \choose \beta(0)} \cdots {\alpha(k-1) \choose \beta(k-1)}$ has a prefix that is accepted by $\aut$. Thus, every continuation of this prefix is in $L$. We conclude that Player~$O$ has a winning strategy for $\delaygame{L}$ whenever $f(0)$ is greater or equal than the depth of $\aut$. 
\end{proof}

To conclude this section, we present a matching lower bound.

\goodbreak
\begin{thm}
	\label{thm_lowerbounds_reachabilitycapsafety}
  For every $ n \ge 0 $,  there is a language~$ L_{n} $ over $\SigmaI = \SigmaO = \set{a,b}$ such that

  \begin{itemize}

  \item $L_n = \lstar{\aut_n}$  for some deterministic acyclic
    automaton~$\aut_n$ with $\size{\aut_n} \in \bigo(n)$ and depth $n+1$,

  \item Player~$ O $ wins $ \delaygame{L_{n}\cdot (\SigmaI \times \SigmaO)^\omega} $ for some constant
    delay function~$ f $, but

  \item Player~$ I $ wins $ \delaygame{L_{n}\cdot (\SigmaI \times \SigmaO)^\omega} $ for every delay
    function~$ f $ with $ f(0) \leq {n} $.

  \end{itemize}
\end{thm}

\begin{proof}
Consider the language
\begin{equation*}
L_n = \left\{ 
{\widx{\alpha}{0} \choose \widx{\beta}{0}} \cdots {\widx{\alpha}{n} \choose \widx{\beta}{n}}   ~\Bigg|~ \beta(0) = \alpha(n)
\right\}.
\end{equation*}
It is straightforward to show that $L_n$ is recognizable by a deterministic acyclic automaton of size~$\bigo(n)$ and depth~$n+1$, and that Player~$O$ wins $\delaygame{L_n \cdot (\SigmaI \times \SigmaO)^\omega}$ if and only if $f(0) > n$.
\end{proof}

Note that all upper bounds in this section hold for non-deterministic automata while the lower bounds hold for deterministic ones.

To conclude this section, let us come back to the lower bounds for reachability and safety conditions presented in Section~\ref{sec_lowerbounds}. The deterministic safety automata~$\aut_n'$ witnessing the exponential lower bound can be equipped with a reachability acceptance that witnesses an exponential lower bound for reachability conditions. However, these automata do not accept the same language. If they would, then Lemma~\ref{lemma_reachabilitycapsafety_charac} yields an acyclic automaton of depth $\bigo(n)$ that represents this language. Thus, Theorem~\ref{thm_reachabilitycapsafetybounds} shows that lookahead~$\bigo(n)$ is sufficient for Player~$O$. However, this contradicts the lower bound proven in Theorem~\ref{thm_lowerboundssafe}.


\section{Conclusion}
\label{sec_conclusion}
We gave the first algorithm that solves $\omega$-regular delay games in exponential time, which is an exponential improvement over the previously known algorithms. We complemented this by showing the problem to be $\exptime$-complete, even for safety conditions. Also, we determined the exact amount of lookahead that is necessary to win $\omega$-regular delay games by proving tight exponential bounds, which already hold for safety and reachability conditions. The $\exptime$-completeness of solving delay games with safety conditions is contrasted by $\pspace$-completeness of solving delay games with reachability conditions. Due to this gap in complexity and due to the exponential lower bounds on the lookahead, we also considered delay games with $\omega$-regular clopen winning conditions. Here, linear lookahead suffices and is in general necessary. Furthermore, determining the winner in such games is $ \conp $-complete. All our lower and upper bounds hold for deterministic automata while for reachability conditions (including clopen conditions) our results even hold for non-deterministic automata. To the best of our knowledge, these are the first non-trivial lower bounds on lookahead and complexity for delay games.

Thus, we completed the picture for deterministic automata and provided partial results for non-deterministic automata, e.g., for automata with reachability acceptance condition. One can trivially obtain upper bounds for the other types of non-deterministic (and universal) automata using determinization, but this incurs an exponential blowup. In current research, we show this to be unavoidable. Furthermore, for alternating automata, we have recently proven tight triply-exponential bounds on the necessary lookahead and shown $\threeexp$-completeness of determining the winner of such a game~\cite{KleinZimmermann16}. Here, the lower bounds already hold for LTL specifications.

An open question concerns the influence on the necessary lookahead and the solution complexity when using different deterministic automata models that recognize the class of $\omega$-regular conditions, e.g., Rabin, Streett, and Muller automata. Indeed, our construction used to prove Theorem~\ref{thm_paritycomplexity_ub} can be adapted to deal with these acceptance conditions, e.g., for conditions given by Muller automata, $\autc$ keeps track of the states visited on a run and $\game(\aut)$ is a Muller game. This yields upper bounds, but it is open whether these are optimal. 

Another promising research direction is the study of restricted classes of strategies for delay games, e.g., finite-state strategies for games with $\omega$-regular winning conditions. Our reduction presented in Section~\ref{subsec_paritycomplex} implies the existence of such strategies, which in general need exponentially many memory states as well as access to the complete lookahead. This result has been proved and a more elegant construction based on the players picking blocks of letters (cf.\ the proof of Lemma~\ref{lem_splitgamecorrectness}), was presented by Salzmann~\cite{Salzmann15}. 

\subsection*{Acknowledgments.}  We thank Bernd Finkbeiner for a
fruitful discussion that lead to Theorem~\ref{thm_lowerboundsreach} and Theorem~\ref{theorem_reachcomplex}. Also, we thank the reviewers for their valuable feedback.


\bibliographystyle{plain}
\bibliography{biblio}


\appendix
\section{Universality of Non-deterministic Reachability Automata}
\label{sec_universality}
We show that universality of non-deterministic reachability automata is $\pspace$-complete. The hardness proof is a small extension of the classical $\pspace$-hardness proof for non-deterministic finite automata~\cite{MeyerStockmeyer72}. We give the proof here for the sake of completeness, as we could not find it in the literature (only some mentionings of the result).

\begin{thm}
The universality problem for non-deterministic reachability automata is $\pspace$-complete.
\end{thm}

\begin{proof}
Membership in $\pspace$ is straightforward: every reachability automaton can be turned into a Büchi automaton turning the accepting states into sinks such that the resulting automaton recognizes the same language\footnote{The correctness relies on the reachability automaton being complete, which we require of our automata.}. Universality for Büchi automata is known to be in $\pspace$~\cite{SistlaVardiWolper85}.

We show hardness by a reduction from the acceptance problem for polynomial space Turing machines. Let $ \tm = (Q, \Sigma, q_I, \delta, \qacc, \qrej)$  be such a machine, where $\delta \colon Q \times \Sigma \rightarrow Q \times \Sigma \times \set{ -1, 0, 1}$ is the transition function. We assume w.l.o.g.\ that the accepting state~$\qacc$ and the rejecting state~$\qrej$ have self-loops. Furthermore, let $x \in \Sigma^*$ be an input for $\tm$.

Let $p$ be a polynomial that bounds the space-consumption of $\tm$. From $p$ we can compute a polynomial~$p'$ such that $2^{p'}$ bounds the time-consumption of $\tm$. Thus, $s = p(|x|)$ and $t = 2^{p'(|x|)}$ are upper bounds on the space- and time-consumption of $\tm$ on $x$, respectively.

We construct a reachability automaton~$\aut$ such that $x$ is accepted by $\tm$ if and only if $\aut$ is not universal. This suffices to prove our claim, since $\pspace$ is closed under complement.

A configuration of $\tm$'s run on $x$ is encoded by a word~$c \in (Q \cup \Sigma)^{s+1}$ as usual. In the following, we do not distinguish between a configuration and its encoding. In the construction, we need to count the configurations of the run. To this end, let $\bin{n} \in \set{0,1}^*$ denote the binary encoding of $n$ in the range~$0 \le n \le t-1$ using $\log_2(t)$ many bits. Now, consider the following three $\omega$-languages over the alphabet $\Sigma' = Q \cup \Sigma \cup \set {0,1,\dollar, \sharpsym}$, where we assume $(Q \cup \Sigma) \cap \set {0,1,\dollar, \sharpsym} = \emptyset$.

\begin{enumerate}
	\item $L_1$ is the set of $\omega$-words that do \emph{not} start with $\bin{0} \dollar\, c_0\, \sharpsym\, $, where $c_0$ is the initial configuration of $\tm$ on $x$.
	
	\item $L_2$ is the set of $\omega$-words having an infix~$\bin{n}\,\dollar\, c\, \sharpsym$ such that
	\begin{itemize}
		\item $n < t-1$,
		\item the infix is \emph{not} followed by $\bin{n+1} \, \dollar \, c' \, \sharpsym$, where $c'$ is the successor configuration of $c$.
	\end{itemize}
	
	\item $L_3$ is the set of $\omega$-words containing the infix~$\bin{t-1}$ such that the first occurrence of the infix~$\bin{t-1}$ is \emph{not} followed by $\dollar \, c \, \sharpsym$ for some accepting configuration~$c$ of length~$s$.
\end{enumerate}

We claim that all three languages can be accepted by non-deterministic reachability automata of polynomial size in $\size{\tm} + |x|$. This is straightforward for $L_1$ and $L_3$ and for $L_2$ we use the fact that it suffices to find a single bit or tape-cell that is not updated correctly. Let $\aut$ be an automaton such that $\lexists{\aut} = L_1 \cup L_2 \cup L_3$, again of polynomial size in $\size{\tm} + |x|$. Such an automaton exists, since the union of  reachability languages can be recognized by the disjoint union of the corresponding automata with a fresh initial state. We claim that $\aut$ has the desired properties.

Assume $\tm$ accepts $x$ and let $c_0, c_1, \ldots, c_k$ for $k \le t-1$ be the accepting run. Consider the word
\[ 
\bin{0} \, \dollar \, c_0 \, \sharpsym \,
\bin{1} \, \dollar \, c_1 \, \sharpsym \,
 \cdots
 \bin{k} \, \dollar \, c_k \, \sharpsym \,
\bin{k+1} \, \dollar \, c_k \, \sharpsym \,
 \cdots
 \bin{t-1} \, \dollar \, c_k \, \sharpsym \,
 0^\omega,
\]
which is not in $\lexists{\aut}$. Hence, $\aut$ is not universal.

Now, assume $\aut$ is not universal. Then, there is a word~$w$ that is not in $\lexists{\aut} = L_1 \cup L_2 \cup L_3$. Thus, $w$ starts with $\bin{0} \dollar\, c_0\, \sharpsym\, $, where $c_0$ is the initial configuration of $\tm$ on $x$. This is followed by $\bin{1} \dollar\, c_1\, \sharpsym\, $, where $c_1$ is the successor configuration of $c_0$. This is continued until we reach an infix~$\bin{t-1} \dollar\, c_{t-1}\, \sharpsym\, $, which is the first occurrence of the infix~$\bin{t-1}$. Due to $w$ not being in $L_3$, we conclude that $c_{t-1}$ is an accepting configuration. Therefore, $\tm$ accepts $x$, as we have constructed an accepting run.
\end{proof}


\end{document}